\documentclass[10pt,letterpaper]{article}
\usepackage[margin=1in]{geometry}
\usepackage{ifpdf}
\ifpdf
    \usepackage[pdftex]{graphicx}
    \graphicspath{{figs/}{matlab/}}   
    \usepackage[update]{epstopdf}
\else
	\usepackage{graphicx}
	\graphicspath{{figs/}{matlab/}}   
\fi

\usepackage{tikz}
\usepackage{pgfplots}
\usepackage{mathrsfs}  
\usepackage{amsmath,amssymb,amsthm,bm,dsfont}  
\usepackage{enumerate}  
\usepackage{color}
\usepackage{setspace}
\usepackage[hidelinks]{hyperref} 
\usepackage{cleveref} 
\usepackage{cite}
\usepackage{subcaption}
\allowdisplaybreaks[2]  
\usepackage{multirow}
\usepackage{cases}
\usepackage{balance}

\usepackage{comment}
\usepackage{xcolor}
\newif\ifEnv

\Envtrue
\ifEnv
    
\else
    \excludecomment{Env}
\fi

\usepackage{optidef}

\newcounter{MYtempeqncnt}
\newtheorem{theorem}{\textbf{Theorem}}
\newtheorem{prop}{\textbf{Proposition}}

\newtheorem{remark}{\textbf{Remark}}

\newtheorem{lemma}{\textbf{Lemma}}

\newcommand{\Fc}{\mathcal{F}}
\newcommand{\Gc}{\mathcal{G}}

\newcommand{\Lc}{\mathcal{L}}

\newcommand{\Pc}{\mathcal{P}}
\newcommand{\Qc}{\mathcal{Q}}

\newcommand{\Tc}{\mathcal{T}}

\newcommand{\Xc}{\mathcal{X}}

\newcommand{\Pb}{\mathbb{P}}

\newcommand{\Eb}{\mathbb{E}}

	\newcommand{\cF}{\mathcal{F}}

	\newcommand{\cP}{\mathcal{P}}
	\newcommand{\cQ}{\mathcal{Q}}

\usepackage{xparse}
\NewDocumentCommand{\ml}{m O{M} O{Q_n^{[M]}}} {\mathcal{L}_{#1}(#2\rightarrow #3)}

\begin{document}

\title{Weakly Private Information Retrieval with Heterogeneity in Servers' Trustfulness}
\author{Wenyuan Zhao, Yu-Shin Huang, Ruida Zhou, and Chao Tian}
\maketitle
\date{}

\maketitle

\begin{abstract}
We study the problem of weakly private information retrieval (PIR) when there is heterogeneity in servers' trustworthiness under the maximal leakage (Max-L) metric and mutual information (MI) metric. A user wishes to retrieve a desired message from N non-colluding servers efficiently, such that the identity of the desired message is not leaked in a significant manner; however, some servers can be more trustworthy than others. We propose a code construction for this setting and optimize the probability distribution for this construction. For the Max-L metric, it is shown that the optimal probability allocation for the proposed scheme essentially separates the delivery patterns into two parts: a completely private part that has the same download overhead as the capacity-achieving PIR code, and a non-private part that allows complete privacy leakage but has no download overhead by downloading only from the most trustful server. The optimal solution is established through a sophisticated analysis of the underlying convex optimization problem and a reduction between the homogeneous setting and the heterogeneous setting. For the MI metric, the homogeneous case is studied first for which the code can be optimized with an explicit probability assignment, while a closed-form solution becomes intractable for the heterogeneous case. Numerical results are provided for both cases to corroborate the theoretical analysis.
\end{abstract}


\section{Introduction}
Private information retrieval (PIR) systems \cite{Chor1995}, initiated by Chor et al. in 1995, were motivated by the necessity to safeguard user privacy during information retrieval. In the standard PIR framework, a user aims to retrieve a specific message from $N$ independent servers, each holding a complete set of $K$ messages. The critical challenge here is to prevent the servers from deducing which message the user is accessing. Typically, messages are quite large, making the download from the servers the dominant communication cost. The efficiency of a PIR system is measured in terms of its PIR capacity: the highest possible information bits per downloaded bit, which was definitively established by Sun and Jafar \cite{Sun2017}. Subsequently, an alternative optimal code, known as the TSC code, was introduced by Tian et al. in 2019 \cite{tian2019capacity}, featuring the smallest possible message length and query set. Many variations and extensions of the canonical PIR problem have been studied, such as PIR with colluding servers\cite{t1,t2,zhou2022two}, storage-constrained PIR \cite{c1,c2,c3,tian2023shannon,c5,c6,c7,c8,c9,c10,zhu2019new,vardy2023private}, PIR with symmetric privacy requirement \cite{d1,d2,d3}, and PIR with side information \cite{s1,s2,s3,s4,s5,s6,s7,lu2023single}; a more comprehensive literature survey can be found in \cite{ulukus2022private}.

In certain scenarios, absolute privacy may not be a critical requirement. In these cases, users might tolerate a scenario in which servers can guess the identity of the desired message with low confidence. This leads to the concept of Weakly Private Information Retrieval (W-PIR), which imposes a more relaxed privacy constraint \cite{Asonov2002,Toledo2016,Samy2019,ZhuqingJia2019,Lin2019,Zhou2020a,Lin2021,Samy2021,lin2021capacity,yakimenka2022optimal}. In exchange for loss of privacy, a higher retrieval rate can be achieved. The assessment of privacy leakage in W-PIR has been approached through several different metrics. Differential privacy was used in \cite{Toledo2016,Samy2019}; conditional entropy was used in \cite{ZhuqingJia2019}; mutual information (MI) was used in \cite{Lin2019}; and the maximal leakage (Max-L) metric (see \cite{Issa2020}) was adopted in \cite{Zhou2020a,Lin2021,qian2022improved}. The W-PIR code in \cite{Zhou2020a} was obtained by breaking the uniform distribution on the retrieval patterns in the TSC code, which increasingly favors the direct download pattern in the code as the privacy requirement is relaxed. The improved construction in \cite{qian2022improved} is based on the observation that in the high leakage regime, the direct retrieval of the desired message becomes increasingly efficient, effectively reducing the overall cost of the retrieval process in scenarios with higher permissible leakage levels.

In this work, we consider the setting where there is heterogeneity in the trustfulness of the servers, that is, some of the servers may be more trustworthy than others. To address this, we propose a general version of the improved W-PIR code \cite{qian2022improved}, which we designate as the W-PIR$^{\#}$ code. Our focus is on fine-tuning the probability allocation within this W-PIR$^{\#}$ framework. Our findings reveal that the optimal strategy for the Max-L metric is surprisingly straightforward: it is essentially a probabilistic sharing between the original TSC code and a direct download from the most trustworthy server. This optimal solution emerges from an intricate examination of a convex optimization problem, initially framed in the context of homogeneously trusted servers. A key aspect of our analysis is that the Max-L setting facilitates a natural transition from a homogeneous trust environment to one that is heterogeneous. For the MI metric, an explicit probability allocation is given in the homogeneous case first, which requires a delicate analysis of the property of the entropy and mutual information. However, a closed-form solution becomes intractable for the heterogeneous setting under the MI metric, for which numerical results are provided to corroborate the theoretical analysis.

The remainder of the paper is organized as follows. In Section \ref{sec:preliminaries}, we introduce the notation and briefly review the maximal leakage metric, the mutual information metric, and the TSC scheme. The proposed optimal code scheme and a motivating example are presented in section \ref{sec:code}. Sections \ref{sec:maxL} and \ref{sec:MI} are devoted to the main results under the Max-L and MI settings, respectively, and Section \ref{sec:conclusion} concludes this paper.

\section{Preliminaries}
\label{sec:preliminaries}
We present a formal introduction to the Weakly Private Information Retrieval (W-PIR) problem, considering both Max-L and MI metrics, with homogeneously and heterogeneously trusted servers. Then we review the TSC code proposed in \cite{tian2019capacity}, followed by a discussion on a particular variant of this code tailored to construct W-PIR codes.

\subsection{Information Retrieval Systems}
There are a total of $K$ mutually independent messages $W_{1:K} := (W_1, W_2, \ldots, W_K)$, each of which consists of $L$ i.i.d. symbols uniformly distributed in a finite set $\Xc$:
\begin{align*}
    & W_k := (W_k[1], W_k[2], \ldots, W_k[L]), ~k\in [1:K], \\
    & L = H(W_1) = H(W_2) = \cdots = H(W_K),
\end{align*}
where $[1:K]$ is used to denote the set $\{ k: k=1,2,\ldots,K\}$ in the sequel, and the entropy is calculated using the logarithm of base $|\Xc|$. Without loss of generality, we assume $K\geq 2$, and an independent copy of all the messages $W_{1:K}$ is stored in a total of $N$ non-colluding servers, respectively. An information retrieval code comprises specific component functions of queries, answers, and decoders. When a user seeks to retrieve a message $W_k$, $k \in [1:K]$, from $N$ servers without disclosing the identity of $k$ to any individual server, a private random key $F^{\ast} \in \Fc^{\ast}$ is firstly used to generate the queries
\begin{align}
Q^{[k]}_n := \phi_n(k, F^{\ast}), \quad n \in [1:N],
\end{align}
where $Q^{[k]}_n \in \Qc_n$, the union of all possible queries $Q^{[k]}_n$ over all $k$. Upon receiving $Q^{[k]}_n$, the server-$n$ generates an answer $A^{(q)}_n$ as a function of the query $Q^{[k]}_n$ and stored messages $W_{1:K}$, which is produced as
\begin{align}
    A^{(q)}_n := \varphi_n(q, W_{1:K}), \quad n \in [1:N]. \label{eqn:varphi}
\end{align} 
In this work, we assume that the answer symbols are distributed in a finite set $\mathcal{Y}$ which is the same as message symbols, i.e., $\mathcal{X}=\mathcal{Y}$. The length of the answer is denoted as $\ell_n^{(q)}$, which may vary according to the query but not to the messages. 

For simplicity of notation, we denote $A^{(Q^{[k]}_{n})}_n$ by $A^{[k]}_n$ and $\ell_n^{(Q^{[k]}_{n})}$ as $\ell_n^{[k]}$, both of which are random variables. Using all the answers from $N$ servers, the user intends to decode the message $\hat{W}_k$ with the function
\begin{align}
    \hat{W}_k := \psi( A^{[k]}_{1:N}, k, F^*).
\end{align}
An information retrieval code is said to be valid only when the desired message is recovered accurately, that is, $\hat{W}_k = W_k$. 

The normalized (worst-case) average download cost is then defined as
\begin{align}
    D := \max_{k \in [1:K]} \Eb\left[\frac{1}{L} \sum_{n = 1}^N \ell_n^{[k]} \right],
\end{align}
where the expectation is taken over the set of all random keys $F^*$. Note that $D$ is determined solely by queries and query distribution, without being influenced by the realization of specific messages or the selection of the desired message index $k$.

\subsection{Maximal Leakage and MI Leakage Metric}
We consider \emph{weakly private information retrieval} (W-PIR) where the user intends to control the amount of information leakage that a database can infer from queries. The privacy leakage in the identity of the desired message, represented as the index $M$ of $W_M$, is due to the query $Q^{[M]}_n$ sent by the user and must be measured for control. We adopted two metrics to study this leakage: Max-L and MI metric.

\vspace{0.1cm}
\noindent{\it The Max-L metric $\Lc(M \rightarrow Q^{[M]}_n)$:} The Max-L leakage on the message $M$ to the $n$-th server is defined in \cite{Zhou2020a} and \cite{Issa2020} as
\begin{align}
	\Lc(M\rightarrow Q_n^{[M]}) = \log_2\bigg{(}\sum_{q \in \Qc_n} \max_{k \in [1:K]} \Pb(Q^{[k]}_n = q )\bigg{)}, \label{eq:max-leakage-def}
\end{align}
which in fact does not depend on the probability distribution of $M$. The server can estimate $M$ more accurately from $Q^{[M]}_n$ in the sense that $\Lc(M\rightarrow Q_n^{[M]})$ is large; on the other hand, the retrieval is completely private when $\Lc(M\rightarrow Q_n^{[M]})=0$, and the distribution of $Q^{[k]}_n$ and $Q^{[k']}_n$ is identical for any $k,k'\in [1:K]$. In \cite{Issa2020}, it was suggested that the exponential of leakage can also be viewed as a measure of information leakage. We can measure the overall privacy leakage under the Max-L metric as $\rho_{(\text{Max-L})}$ by the weighted sum of the exponential leakage amounts to the individual servers
\begin{align*}
\rho_{(\text{Max-L})}\triangleq \sum_{n=1}^N \gamma_n2^{\Lc(M \rightarrow Q^{[M]}_n)},
\end{align*}
where $\gamma_n>0$. Here, the exponential is taken for simplifying the analysis. Without loss of generality, we assume $\gamma_1\leq \gamma_2\leq\ldots\leq \gamma_N$. In most cases, we shall choose to normalize the weights so that $\sum_{n=1}^N \gamma_n=1$, however, this is not critical, and we shall utilize this fact in subsequent derivations. Note that the $\gamma_n$'s are strictly positive, as otherwise the problem becomes trivial since the optimal strategy is to directly retrieve all messages from this completely trusted server. For the homogeneous trust setting, we simply set $\gamma_n=1/N$. 

\vspace{0.1cm}
\noindent{\it The MI metric $\text{MI}(M\rightarrow Q_n^{[M]})$:} The MI leakage on the message $M$ to the $n$-th server is defined as 
\begin{align}
	\text{MI}(M\rightarrow Q_n^{[M]}) := I(M;Q_n), \label{eq:MI}
\end{align}
where $Q_n$ is the random variable which represents the query induced jointly by the random key $F^*$ and the random message index $M$. We assume the message index $M$ is uniformly distributed in $[1:K]$ in this work, and the overall leakage under the MI metric is similarly defined as the weighted sum of the leaked MI to the individual servers
\begin{align}
\rho_{(\text{MI})}\triangleq \sum_{n=1}^N \gamma_n\text{MI}(M\rightarrow Q_n^{[M]}).
\end{align}

A \emph{valid} code for W-PIR with $K$ messages and $N$ servers under the download cost constraint $d$ is a collection of functions $(\{\phi_n\}_{n \in [1:N]},\{\varphi_n\}_{n \in [1:N]},\psi)$ that can correctly retrieve the desired message and additionally satisfies the download constraint $D\leq d$. A leakage $\rho$ is called achievable for the download cost constraint $d$, if there exists a valid code such that the leakage $\Lc\leq \rho$ under the download constraint $d$. The closure of the collection of these $(\rho,d)$ pairs is called the achievable $(\rho,d)$ region, denoted by $\mathcal{G}_{\text{Max-L}}$ and $\mathcal{G}_{\text{MI}}$ under the Max-L metric and the MI metric, respectively.

\subsection{The TSC Code and its Permuted Variant}
We briefly review the TSC code in \cite{tian2019capacity}, and propose one permuted variant of this code that is suitable for our cases where limited information leakage is allowed. In the TSC code, each message consists of symbols $L = N-1$, i.e., 
\begin{align*}
    W_k=(W_k[1], W_k[2], \ldots, W_k[N-1]), ~k\in [1:K],
\end{align*}
where a dummy symbol $W_k[0]=0$ is prepended at the beginning of all messages. To better facilitate the construction of the new code, especially in a heterogeneous setting, we provide a variation of the original construction, which can be viewed as probabilistic sharing between all permutations (over $N$ servers) of the PIR code in \cite{tian2019capacity}. 

Define a random vector $F\triangleq(F_1, F_2, \ldots, F_{K-1})$ of length-$(K-1)$, which is uniformly distributed in $[0:N-1]^{K-1}$ and a random bijective mapping $\pi: [1:N]\rightarrow [0:N-1]$ (a permutation on the set $[1:N]$ but downshifted by $1$). The random key $F^*$ is defined as the concatenation of $F$ and $\pi$: 
\begin{align}
    F^*: = (F,\pi)=(F_1,F_2,\dots,F_{K-1},\pi),
\end{align} 
where $F_1,\dots,F_{K-1},\pi$ are mutually independent random variables, and the distribution of $\pi$ will be specified later. We shall use $f$ to denote a specific realization of the random vector $F$, and use $\mathcal{F}$ to denote the set of $[0:N-1]^{K-1}$, i.e., the set of possible values of the partial random key $F$.

The function $\phi_n^*(k, F^*)$ to generate the query $Q_n^{[k]}$ for retrieving message-$k$ from server-$n$ is defined as 
\begin{equation}
\begin{aligned}
     \phi_n^*(k, F^*) \triangleq \bigg( & F_1, F_2, \dots,F_{k-1}, \big(\pi(n)-\sum_{j = 1}^{K-1}F_j \big)_N, F_{k}, F_{k+1}, \dots, F_{K-1} \bigg), \label{eqn:tsc-phi}
\end{aligned}
\end{equation}
where $(\cdot)_N$ denotes the modulo $N$ operation. Each query is a length-$K$ vector, and we denote the $m$-th symbol of $Q_{n}^{[k]}$ as $Q_{n,m}^{[k]}$. Upon receiving the query, the answer $A_n^{[k]}$ is returned from server-$n$ generated by the function
\begin{align}
    \varphi^{\ast}_n(q,W_{1:K})& \triangleq W_1[Q_{n,1}^{[k]}]\oplus W_2[Q_{n,2}^{[k]}]\oplus \cdots \oplus W_K[Q_{n,K}^{[k]}]\notag\\
    & = W_k[(\pi(n)-\sum_{j = 1}^{K-1}F_j)_N] \oplus \mathscr{I},
\end{align}
where $\oplus$ is addition in the given finite field. As a result, $\mathscr{I}$ is the interference signal given by
\begin{equation}
    \begin{split}
        \mathscr{I} = W_1[F_1]\oplus\cdots\oplus W_{k-1}[F_{k-1}] \oplus W_{k+1}[F_{k}] \oplus \cdots \oplus W_{K}[F_{K-1}].
    \end{split}
\end{equation}
Given the existence of an $n^{\ast} \in [1:N]$ such that $(\pi(n^{\ast})-\sum_{j = 1}^{K-1}F_j)_N = 0$, it implies that $A_{n^{\ast}}^{[k]} = \mathscr{I}$. Consequently, the user can retrieve the desired message $W_k$ by subtracting $\mathscr{I}$ from $A_{n}^{[k]}$ for all $n \neq n^{\ast}$. It should be noted that with a probability of $N^{-(K-1)}$, the interference signal $\mathscr{I}$ may only consist of dummy symbols, which eliminates the need for its download entirely. In such cases, direct retrieval can be achieved by accessing the desired message from $N-1$ servers, one symbol per server. The download cost is therefore 
\begin{align}
    D^{\ast} = \frac{N}{N-1}\left(1-\frac{1}{N^{K-1}}\right) + \frac{1}{N^{K-1}} =\frac{ 1 - N^{-K} }{ 1 - N^{-1} },
\end{align}
matching the capacity result given in \cite{Sun2017}. It can be shown that there is no privacy leakage regardless of the distribution of the random permutation $\pi$, since for each fixed permutation the resultant code is private. An example of the code (with adjusted probabilities for W-PIR) is given in Section \ref{sec:code} (Tables \ref{tab:N3K2}, the lower halves without the $\#$ parts); more details can be found in \cite{tian2019capacity}.

\subsection{Weakly PIR: Reassigned Probabilities in TSC}
In the permuted variant of the generalized TSC code, we can reduce the download cost by assigning a higher probability to random keys when $F_1=F_2=\ldots=F_{K-1}=0$, i.e. the pattern for which the retrieval downloads the messages without interference at the cost of $L$. If the probabilities of these random keys are very high, then the messages will be more likely to be downloaded directly from the $N-1$ servers, resulting in privacy leakage but a lower download cost. If the probabilities of these random keys are the same as for all other keys, then we have the original permuted variant of the TSC code, resulting in completely private retrieval. By adjusting these probability assignments, we obtain a range of weakly private information retrieval codes achieving different trade-offs between the download cost and the privacy leakage. Almost all existing W-PIR code constructions essentially utilize this approach in some manner \cite{Lin2019,Zhou2020a,Lin2021,Samy2021}.

\section{W-PIR$^\#$: Generalized TSC Code with Escape Retrieval Symbols} 
\label{sec:code}

\subsection{The Proposed Code Construction}

For high-leakage situations, the weakly private information retrieval code given above by reassigning probabilities in the TSC code does not perform well. To see this, consider the extreme case of the minimum download cost point, this code will download the messages directly from $N-1$ servers, resulting in a privacy leakage to all these servers. However, we can instead directly download the message from a single server, therefore, leaking the message index to only one server. This motivates the addition of such direct download patterns in our proposed new code, and these download patterns are denoted as $\#$. 

We next present the W-PIR$^\#$ code, which is essentially a probabilistic sharing scheme between the generalized TSC code and the direct retrieval patterns from individual servers. In this code, we again set $L=N-1$. The random key $F^*$ is generated from the set $\Fc^*$ with a probability distribution $\Pb_k(F^*)$, where $\Fc^*= ([0:N-1]^{K-1}\times \cP) \cup [1:N]$ for which $\cP=\{\pi\}$ is the set of all bijective mappings $[1:N]\rightarrow [0:N-1]$. This probability distribution of requesting message $k$ is denoted as 
\begin{align}
    \Pb_k(F^*) = \begin{cases} p^{k,F^*}_{(\#)}, & F^* \in [1:N]\\
    p^{k,\pi}_{(f)}, & F^* = (f,\pi)\in [0:N-1]^{K-1}\times \cP
    \end{cases},
\end{align}
which needs to satisfy
\begin{align}
    \sum_{n=1}^N p^{k,n}_{(\#)} + \sum_{f\in \mathcal{F}}\sum_{\pi\in \cP} p^{k,\pi}_{(f)}  = 1,\quad k=1,2,\ldots,K.
\end{align}

The query $Q_n^{[k]}$ to server-$n$ is produced as:
\begin{align}
    Q_n^{[k]} = \begin{cases} \#_k, & F^* = n\\ 
    \underline{0_K}, & F^* \in [1:N],~ F^* \neq n\\ 
    \phi_{n}^*(k,F^*), & F^*\notin [1:N]\end{cases},
\end{align}
where $\underline{0_K}$ is the length-$K$ all-zero vector. The answer $A_n^{[k]}$ from server-$n$ is generated as
\begin{align}
    A_n^{[k]} = \begin{cases} W_k, & q = \#_k\\ 
    \varphi^*(q,W_{1:K}), &  q \notin \{\#_k: k\in[1:K]\}.\end{cases}
\end{align}
The decoding procedure follows directly from the original generalized TSC code when $F^* \notin [1:N]$ and does not require decoding when $F^* \in [1:N]$. We will refer to this code as W-PIR$^{\#}$. A simpler version of W-PIR$^{\#}$, which does not allow all permutations, was first presented in \cite{qian2022improved}.

The correctness of the code is obvious, and the download cost $D$ can be simply computed as 
\begin{align}
    p_d^k &\triangleq \sum_{n=1}^N p^{k,n}_{(\#)} + \sum_{\pi\in \cP} p^{k,\pi}_{(\underline{0_{K-1}})}, k\in [1:K],\label{eq:pd}\\
    D & = \max_k \left(p^k_d + \frac{N}{N-1}(1-p^k_d)\right),
\end{align}
where $p_d^k$ is the overall probability of using a direct download to retrieve message $k$, either by retrieving from $(N-1)$ servers, or by retrieving from only $1$ servers. We defer the analysis of privacy to Sections \ref{sec:maxL} and \ref{sec:MI}.

\subsection{An Example When $K=2$ and $N=3$}

We provide an example of more explicit code construction in Table \ref{tab:N3K2}. Consider the case with $K=2$ messages and $N=3$ servers. Here the dummy symbols $a_0$ and $b_0$ are omitted for conciseness. The random key $F^*$ has a total of $21$ possible realizations, each is either associated with a random key and the downshifted permutation function $(F,\pi)$, or a direct retrieval specified by the server index.  The message length is $L=N-1=2$, and we write the two messages $W_1 = (a_1,a_2)$, $W_2 = (b_1,b_2)$. The queries and answers are given in Table \ref{tab:N3K2}. The queries in the top three rows of the two tables directly request the full message from a single server denoted by $\#_1$ and $\#_2$, while the remaining rows are essentially the TSC code with different probabilities for the queries, assigned according to their interference signals and the permutation functions. Note that the interference signal is controlled by the first $(K-1)$ entries $F$ of the random key $F^*$. Let us denote $\|F\|$ as the \emph{size} of the interference corresponding to the random key $F$, which is also its Hamming weight; in this example, $\|F\|$ can only be $0$ or $1$. Note that $|\mathcal{F}|$ is used to denote the \emph{cardinality} of the set $\mathcal{F}$, different from $\|F\|$.

\begin{table*}[tb!]
\caption{Proposed code for $N=3, K=2$}
\label{tab:N3K2}
    \centering
    \begin{subtable}[t]{0.495\textwidth}
        \caption{Retrieval of $W_{1}$}
        \centering
        \resizebox{\linewidth}{!}{
        \renewcommand{\arraystretch}{1.2}
        \begin{tabular}{|ccccccccc|}
        \hline
        \multicolumn{9}{|c|}{Requesting Message $k=1$}                                                                                                                                                                                                                  \\ \hline
        \multicolumn{1}{|c|}{\multirow{2}{*}{Prob.}} & \multicolumn{1}{c|}{\multirow{2}{*}{$F$}} & \multicolumn{1}{c|}{\multirow{2}{*}{$\pi$ or $n$}}&\multicolumn{2}{c|}{Server 1}                                         & \multicolumn{2}{c|}{Server 2}                                         & \multicolumn{2}{c|}{Server 3}                    \\ \cline{4-9} 
        \multicolumn{1}{|c|}{}                       & \multicolumn{1}{c|}{}      & \multicolumn{1}{c|}{}               & \multicolumn{1}{c|}{$Q_{1}^{[1]}$} & \multicolumn{1}{c|}{$A_{1}$}     & \multicolumn{1}{c|}{$Q_{2}^{[1]}$} & \multicolumn{1}{c|}{$A_{2}$}     & \multicolumn{1}{c|}{$Q_{3}^{[1]}$} & $A_{3}$     \\ \hline\hline
        \multicolumn{1}{|c|}{$p_{(\#)}^{1,1}$}     & \multicolumn{1}{c|}{$\#$} & \multicolumn{1}{c|}{$1$}               & \multicolumn{1}{c|}{$\#_1$}        & \multicolumn{1}{c|}{$a_1,a_2$}   & \multicolumn{1}{c|}{$00$}          & \multicolumn{1}{c|}{$\emptyset$} & \multicolumn{1}{c|}{$00$}          & $\emptyset$ \\ \hline
        \multicolumn{1}{|c|}{$p_{(\#)}^{1,2}$}     & \multicolumn{1}{c|}{$\#$}  & \multicolumn{1}{c|}{$2$}               & \multicolumn{1}{c|}{$00$}          & \multicolumn{1}{c|}{$\emptyset$} & \multicolumn{1}{c|}{$\#_1$}        & \multicolumn{1}{c|}{$a_1,a_2$}   & \multicolumn{1}{c|}{$00$}          & $\emptyset$ \\ \hline
        \multicolumn{1}{|c|}{$p_{(\#)}^{1,3}$}     & \multicolumn{1}{c|}{$\#$} & \multicolumn{1}{c|}{$3$}                & \multicolumn{1}{c|}{$00$}          & \multicolumn{1}{c|}{$\emptyset$} & \multicolumn{1}{c|}{$00$}          & \multicolumn{1}{c|}{$\emptyset$} & \multicolumn{1}{c|}{$\#_1$}        & $a_1,a_2$   \\ \hline\hline
        \multicolumn{1}{|c|}{$p_{(0)}^{1,{[2,1,0]}}$}    & \multicolumn{1}{c|}{$0$}  & \multicolumn{1}{c|}{$(2,1,0)$}                & \multicolumn{1}{c|}{$20$}          & \multicolumn{1}{c|}{$a_2$}       & \multicolumn{1}{c|}{$10$}          & \multicolumn{1}{c|}{$a_1$}       & \multicolumn{1}{c|}{$00$}          & $\emptyset$ \\ \hline
        \multicolumn{1}{|c|}{$p_{(0)}^{1,[2, 0, 1]}$}    & \multicolumn{1}{c|}{$0$}  & \multicolumn{1}{c|}{$(2,0,1)$}                & \multicolumn{1}{c|}{$20$}          & \multicolumn{1}{c|}{$a_2$}       & \multicolumn{1}{c|}{$00$}          & \multicolumn{1}{c|}{$\emptyset$} & \multicolumn{1}{c|}{$10$}          & $a_1$       \\ \hline
        \multicolumn{1}{|c|}{$p_{(0)}^{1,[1,2,0]}$}    & \multicolumn{1}{c|}{$0$}    & \multicolumn{1}{c|}{$(1,2,0)$}              & \multicolumn{1}{c|}{$10$}          & \multicolumn{1}{c|}{$a_1$}       & \multicolumn{1}{c|}{$20$}          & \multicolumn{1}{c|}{$a_2$}       & \multicolumn{1}{c|}{$00$}          & $\emptyset$ \\ \hline
        \multicolumn{1}{|c|}{$p_{(0)}^{1,[1,0,2]}$}    & \multicolumn{1}{c|}{$0$}    & \multicolumn{1}{c|}{$(1,0,2)$}              & \multicolumn{1}{c|}{$10$}          & \multicolumn{1}{c|}{$a_1$}       & \multicolumn{1}{c|}{$00$}          & \multicolumn{1}{c|}{$\emptyset$} & \multicolumn{1}{c|}{$20$}          & $a_2$       \\ \hline
        \multicolumn{1}{|c|}{$p_{(0)}^{1,[0,2,1]}$}    & \multicolumn{1}{c|}{$0$}    & \multicolumn{1}{c|}{$(0,2,1)$}              & \multicolumn{1}{c|}{$00$}          & \multicolumn{1}{c|}{$\emptyset$} & \multicolumn{1}{c|}{$20$}          & \multicolumn{1}{c|}{$a_2$}       & \multicolumn{1}{c|}{$10$}          & $a_1$       \\ \hline
        \multicolumn{1}{|c|}{$p_{(0)}^{1,[0,1,2]}$}    & \multicolumn{1}{c|}{$0$}    & \multicolumn{1}{c|}{$(0,1,2)$}              & \multicolumn{1}{c|}{$00$}          & \multicolumn{1}{c|}{$\emptyset$} & \multicolumn{1}{c|}{$10$}          & \multicolumn{1}{c|}{$a_1$}       & \multicolumn{1}{c|}{$20$}          & $a_2$       \\ \hline
        \multicolumn{1}{|c|}{$p_{(1)}^{1,[2,1,0]}$}    & \multicolumn{1}{c|}{$1$}    & \multicolumn{1}{c|}{$(2,1,0)$}              & \multicolumn{1}{c|}{$11$}          & \multicolumn{1}{c|}{$a_1\oplus b_1$}   & \multicolumn{1}{c|}{$01$}          & \multicolumn{1}{c|}{$b_1$}       & \multicolumn{1}{c|}{$21$}          & $a_2\oplus b_1$   \\ \hline
        \multicolumn{1}{|c|}{$p_{(1)}^{1,[2,0,1]}$}    & \multicolumn{1}{c|}{$1$}    & \multicolumn{1}{c|}{$(2,0,1)$}              & \multicolumn{1}{c|}{$11$}          & \multicolumn{1}{c|}{$a_1\oplus b_1$}   & \multicolumn{1}{c|}{$21$}          & \multicolumn{1}{c|}{$a_2\oplus b_1$}   & \multicolumn{1}{c|}{$01$}          & $b_1$       \\ \hline
        \multicolumn{1}{|c|}{$p_{(1)}^{1,[1,2,0]}$}    & \multicolumn{1}{c|}{$1$}    & \multicolumn{1}{c|}{$(1,2,0)$}              & \multicolumn{1}{c|}{$01$}          & \multicolumn{1}{c|}{$b_1$}       & \multicolumn{1}{c|}{$11$}          & \multicolumn{1}{c|}{$a_1\oplus b_1$}   & \multicolumn{1}{c|}{$21$}          & $a_2\oplus b_1$   \\ \hline
        \multicolumn{1}{|c|}{$p_{(1)}^{1,[1,0,2]}$}    & \multicolumn{1}{c|}{$1$}    & \multicolumn{1}{c|}{$(1,0,2)$}              & \multicolumn{1}{c|}{$01$}          & \multicolumn{1}{c|}{$b_1$}       & \multicolumn{1}{c|}{$21$}          & \multicolumn{1}{c|}{$a_2\oplus b_1$}   & \multicolumn{1}{c|}{$11$}          & $a_1\oplus b_1$   \\ \hline
        \multicolumn{1}{|c|}{$p_{(1)}^{1,[0,2,1]}$}    & \multicolumn{1}{c|}{$1$}    & \multicolumn{1}{c|}{$(0,2,1)$}              & \multicolumn{1}{c|}{$21$}          & \multicolumn{1}{c|}{$a_2\oplus b_1$}   & \multicolumn{1}{c|}{$11$}          & \multicolumn{1}{c|}{$a_1\oplus b_1$}   & \multicolumn{1}{c|}{$01$}          & $b_1$       \\ \hline
        \multicolumn{1}{|c|}{$p_{(1)}^{1,[0,1,2]}$}    & \multicolumn{1}{c|}{$1$}    & \multicolumn{1}{c|}{$(0,1,2)$}              & \multicolumn{1}{c|}{$21$}          & \multicolumn{1}{c|}{$a_2\oplus b_1$}   & \multicolumn{1}{c|}{$01$}          & \multicolumn{1}{c|}{$b_1$}       & \multicolumn{1}{c|}{$11$}          & $a_1\oplus b_1$   \\ \hline
        \multicolumn{1}{|c|}{$p_{(2)}^{1,[2,1,0]}$}    & \multicolumn{1}{c|}{$2$}    & \multicolumn{1}{c|}{$(2,1,0)$}              & \multicolumn{1}{c|}{$02$}          & \multicolumn{1}{c|}{$b_2$}       & \multicolumn{1}{c|}{$22$}          & \multicolumn{1}{c|}{$a_2\oplus b_2$}   & \multicolumn{1}{c|}{$12$}          & $a_1\oplus b_2$   \\ \hline
        \multicolumn{1}{|c|}{$p_{(2)}^{1,[2,0,1]}$}    & \multicolumn{1}{c|}{$2$}    & \multicolumn{1}{c|}{$(2,0,1)$}              & \multicolumn{1}{c|}{$02$}          & \multicolumn{1}{c|}{$b_2$}       & \multicolumn{1}{c|}{$12$}          & \multicolumn{1}{c|}{$a_1\oplus b_2$}   & \multicolumn{1}{c|}{$22$}          & $a_2\oplus b_2$   \\ \hline
        \multicolumn{1}{|c|}{$p_{(2)}^{1,[1,2,0]}$}    & \multicolumn{1}{c|}{$2$}    & \multicolumn{1}{c|}{$(1,2,0)$}              & \multicolumn{1}{c|}{$22$}          & \multicolumn{1}{c|}{$a_2\oplus b_2$}   & \multicolumn{1}{c|}{$02$}          & \multicolumn{1}{c|}{$b_2$}       & \multicolumn{1}{c|}{$12$}          & $a_1\oplus b_2$   \\ \hline
        \multicolumn{1}{|c|}{$p_{(2)}^{1,[1,0,2]}$}    & \multicolumn{1}{c|}{$2$}    & \multicolumn{1}{c|}{$(1,0,2)$}              & \multicolumn{1}{c|}{$22$}          & \multicolumn{1}{c|}{$a_2\oplus b_2$}   & \multicolumn{1}{c|}{$12$}          & \multicolumn{1}{c|}{$a_1\oplus b_2$}   & \multicolumn{1}{c|}{$02$}          & $b_2$       \\ \hline
        \multicolumn{1}{|c|}{$p_{(2)}^{1,[0,2,1]}$}    & \multicolumn{1}{c|}{$2$}    & \multicolumn{1}{c|}{$(0,2,1)$}              & \multicolumn{1}{c|}{$12$}          & \multicolumn{1}{c|}{$a_1\oplus b_2$}   & \multicolumn{1}{c|}{$02$}          & \multicolumn{1}{c|}{$b_2$}       & \multicolumn{1}{c|}{$22$}          & $a_2\oplus b_2$   \\ \hline
        \multicolumn{1}{|c|}{$p_{(2)}^{1,[0,1,2]}$}    & \multicolumn{1}{c|}{$2$}    & \multicolumn{1}{c|}{$(0,1,2)$}              & \multicolumn{1}{c|}{$12$}          & \multicolumn{1}{c|}{$a_1\oplus b_2$}   & \multicolumn{1}{c|}{$22$}          & \multicolumn{1}{c|}{$a_2\oplus b_2$}   & \multicolumn{1}{c|}{$02$}          & $b_2$       \\ \hline
        \end{tabular}
        }
    \end{subtable}
    \begin{subtable}[t]{0.495\textwidth}
        \caption{Retrieval of $W_{2}$}
        \centering
        \resizebox{\linewidth}{!}{
        \renewcommand{\arraystretch}{1.13}
        \begin{tabular}{|ccccccccc|}
        \hline
        \multicolumn{9}{|c|}{Requesting Message $k=2$}                                                                                                                                                                                                                                               \\ \hline
        \multicolumn{1}{|c|}{\multirow{2}{*}{Prob.}} & \multicolumn{1}{c|}{\multirow{2}{*}{$F$}} & \multicolumn{1}{c|}{\multirow{2}{*}{$\pi$ or $n$}}&\multicolumn{2}{c|}{Server 1}                                         & \multicolumn{2}{c|}{Server 2}                                         & \multicolumn{2}{c|}{Server 3}                    \\ \cline{4-9} 
        \multicolumn{1}{|c|}{}                       & \multicolumn{1}{c|}{}      & \multicolumn{1}{c|}{}               & \multicolumn{1}{c|}{$Q_{1}^{[2]}$} & \multicolumn{1}{c|}{$A_{1}$}     & \multicolumn{1}{c|}{$Q_{2}^{[2]}$} & \multicolumn{1}{c|}{$A_{2}$}     & \multicolumn{1}{c|}{$Q_{3}^{[2]}$} & $A_{3}$     \\ \hline\hline
        \multicolumn{1}{|c|}{$p_{(\#)}^{2,1}$}     & \multicolumn{1}{c|}{$\#$} & \multicolumn{1}{c|}{$1$}               & \multicolumn{1}{c|}{$\#_2$}        & \multicolumn{1}{c|}{$b_1,b_2$}   & \multicolumn{1}{c|}{$00$}          & \multicolumn{1}{c|}{$\emptyset$} & \multicolumn{1}{c|}{$00$}          & $\emptyset$ \\ \hline
        \multicolumn{1}{|c|}{$p_{(\#)}^{2,2}$}     & \multicolumn{1}{c|}{$\#$}  & \multicolumn{1}{c|}{$2$}               & \multicolumn{1}{c|}{$00$}          & \multicolumn{1}{c|}{$\emptyset$} & \multicolumn{1}{c|}{$\#_2$}        & \multicolumn{1}{c|}{$b_1,b_2$}   & \multicolumn{1}{c|}{$00$}          & $\emptyset$ \\ \hline
        \multicolumn{1}{|c|}{$p_{(\#)}^{2,3}$}     & \multicolumn{1}{c|}{$\#$} & \multicolumn{1}{c|}{$3$}                & \multicolumn{1}{c|}{$00$}          & \multicolumn{1}{c|}{$\emptyset$} & \multicolumn{1}{c|}{$00$}          & \multicolumn{1}{c|}{$\emptyset$} & \multicolumn{1}{c|}{$\#_2$}        & $b_1,b_2$   \\ \hline\hline
        \multicolumn{1}{|c|}{$p_{(0)}^{2,{[2,1,0]}}$}    & \multicolumn{1}{c|}{$0$}  & \multicolumn{1}{c|}{$(2,1,0)$}                & \multicolumn{1}{c|}{$02$}          & \multicolumn{1}{c|}{$b_2$}       & \multicolumn{1}{c|}{$01$}          & \multicolumn{1}{c|}{$b_1$}       & \multicolumn{1}{c|}{$00$}          & $\emptyset$ \\ \hline
        \multicolumn{1}{|c|}{$p_{(0)}^{2,[2, 0, 1]}$}    & \multicolumn{1}{c|}{$0$}  & \multicolumn{1}{c|}{$(2,0,1)$}                & \multicolumn{1}{c|}{$02$}          & \multicolumn{1}{c|}{$b_2$}       & \multicolumn{1}{c|}{$00$}          & \multicolumn{1}{c|}{$\emptyset$} & \multicolumn{1}{c|}{$01$}          & $b_1$       \\ \hline
        \multicolumn{1}{|c|}{$p_{(0)}^{2,[1,2,0]}$}    & \multicolumn{1}{c|}{$0$}    & \multicolumn{1}{c|}{$(1,2,0)$}              & \multicolumn{1}{c|}{$01$}          & \multicolumn{1}{c|}{$b_1$}       & \multicolumn{1}{c|}{$02$}          & \multicolumn{1}{c|}{$b_2$}       & \multicolumn{1}{c|}{$00$}          & $\emptyset$ \\ \hline
        \multicolumn{1}{|c|}{$p_{(0)}^{2,[1,0,2]}$}    & \multicolumn{1}{c|}{$0$}    & \multicolumn{1}{c|}{$(1,0,2)$}              & \multicolumn{1}{c|}{$01$}          & \multicolumn{1}{c|}{$b_1$}       & \multicolumn{1}{c|}{$00$}          & \multicolumn{1}{c|}{$\emptyset$} & \multicolumn{1}{c|}{$02$}          & $b_2$       \\ \hline
        \multicolumn{1}{|c|}{$p_{(0)}^{2,[0,2,1]}$}    & \multicolumn{1}{c|}{$0$}    & \multicolumn{1}{c|}{$(0,2,1)$}              & \multicolumn{1}{c|}{$00$}          & \multicolumn{1}{c|}{$\emptyset$} & \multicolumn{1}{c|}{$02$}          & \multicolumn{1}{c|}{$b_2$}       & \multicolumn{1}{c|}{$01$}          & $b_1$       \\ \hline
        \multicolumn{1}{|c|}{$p_{(0)}^{2,[0,1,2]}$}    & \multicolumn{1}{c|}{$0$}    & \multicolumn{1}{c|}{$(0,1,2)$}              & \multicolumn{1}{c|}{$00$}          & \multicolumn{1}{c|}{$\emptyset$} & \multicolumn{1}{c|}{$01$}          & \multicolumn{1}{c|}{$b_1$}       & \multicolumn{1}{c|}{$02$}          & $b_2$       \\ \hline
        \multicolumn{1}{|c|}{$p_{(1)}^{2,[2,1,0]}$}    & \multicolumn{1}{c|}{$1$}    & \multicolumn{1}{c|}{$(2,1,0)$}              & \multicolumn{1}{c|}{$11$}          & \multicolumn{1}{c|}{$a_1\oplus b_1$}   & \multicolumn{1}{c|}{$10$}          & \multicolumn{1}{c|}{$a_1$}       & \multicolumn{1}{c|}{$12$}          & $a_1\oplus b_2$   \\ \hline
        \multicolumn{1}{|c|}{$p_{(1)}^{2,[2,0,1]}$}    & \multicolumn{1}{c|}{$1$}    & \multicolumn{1}{c|}{$(2,0,1)$}              & \multicolumn{1}{c|}{$11$}          & \multicolumn{1}{c|}{$a_1\oplus b_1$}   & \multicolumn{1}{c|}{$12$}          & \multicolumn{1}{c|}{$a_1\oplus b_2$}   & \multicolumn{1}{c|}{$10$}          & $a_1$       \\ \hline
        \multicolumn{1}{|c|}{$p_{(1)}^{2,[1,2,0]}$}    & \multicolumn{1}{c|}{$1$}    & \multicolumn{1}{c|}{$(1,2,0)$}              & \multicolumn{1}{c|}{$10$}          & \multicolumn{1}{c|}{$a_1$}       & \multicolumn{1}{c|}{$11$}          & \multicolumn{1}{c|}{$a_1\oplus b_1$}   & \multicolumn{1}{c|}{$12$}          & $a_1\oplus b_2$   \\ \hline
        \multicolumn{1}{|c|}{$p_{(1)}^{2,[1,0,2]}$}    & \multicolumn{1}{c|}{$1$}    & \multicolumn{1}{c|}{$(1,0,2)$}              & \multicolumn{1}{c|}{$10$}          & \multicolumn{1}{c|}{$a_1$}       & \multicolumn{1}{c|}{$12$}          & \multicolumn{1}{c|}{$a_1\oplus b_2$}   & \multicolumn{1}{c|}{$11$}          & $a_1\oplus b_1$   \\ \hline
        \multicolumn{1}{|c|}{$p_{(1)}^{2,[0,2,1]}$}    & \multicolumn{1}{c|}{$1$}    & \multicolumn{1}{c|}{$(0,2,1)$}              & \multicolumn{1}{c|}{$12$}          & \multicolumn{1}{c|}{$a_1\oplus b_2$}   & \multicolumn{1}{c|}{$11$}          & \multicolumn{1}{c|}{$a_1\oplus b_1$}   & \multicolumn{1}{c|}{$10$}          & $a_1$       \\ \hline
        \multicolumn{1}{|c|}{$p_{(1)}^{2,[0,1,2]}$}    & \multicolumn{1}{c|}{$1$}    & \multicolumn{1}{c|}{$(0,1,2)$}              & \multicolumn{1}{c|}{$12$}          & \multicolumn{1}{c|}{$a_1\oplus b_2$}   & \multicolumn{1}{c|}{$10$}          & \multicolumn{1}{c|}{$a_1$}       & \multicolumn{1}{c|}{$11$}          & $a_1\oplus b_1$   \\ \hline
        \multicolumn{1}{|c|}{$p_{(2)}^{2,[2,1,0]}$}    & \multicolumn{1}{c|}{$2$}    & \multicolumn{1}{c|}{$(2,1,0)$}              & \multicolumn{1}{c|}{$20$}          & \multicolumn{1}{c|}{$a_2$}       & \multicolumn{1}{c|}{$22$}          & \multicolumn{1}{c|}{$a_2\oplus b_2$}   & \multicolumn{1}{c|}{$21$}          & $a_2\oplus b_1$   \\ \hline
        \multicolumn{1}{|c|}{$p_{(2)}^{2,[2,0,1]}$}    & \multicolumn{1}{c|}{$2$}    & \multicolumn{1}{c|}{$(2,0,1)$}              & \multicolumn{1}{c|}{$20$}          & \multicolumn{1}{c|}{$a_2$}       & \multicolumn{1}{c|}{$21$}          & \multicolumn{1}{c|}{$a_2\oplus b_1$}   & \multicolumn{1}{c|}{$22$}          & $a_2\oplus b_2$   \\ \hline
        \multicolumn{1}{|c|}{$p_{(2)}^{2,[1,2,0]}$}    & \multicolumn{1}{c|}{$2$}    & \multicolumn{1}{c|}{$(1,2,0)$}              & \multicolumn{1}{c|}{$22$}          & \multicolumn{1}{c|}{$a_2\oplus b_2$}   & \multicolumn{1}{c|}{$20$}          & \multicolumn{1}{c|}{$a_2$}       & \multicolumn{1}{c|}{$21$}          & $a_2\oplus b_1$   \\ \hline
        \multicolumn{1}{|c|}{$p_{(2)}^{2,[1,0,2]}$}    & \multicolumn{1}{c|}{$2$}    & \multicolumn{1}{c|}{$(1,0,2)$}              & \multicolumn{1}{c|}{$22$}          & \multicolumn{1}{c|}{$a_2\oplus b_2$}   & \multicolumn{1}{c|}{$21$}          & \multicolumn{1}{c|}{$a_2\oplus b_1$}   & \multicolumn{1}{c|}{$20$}          & $a_2$       \\ \hline
        \multicolumn{1}{|c|}{$p_{(2)}^{2,[0,2,1]}$}    & \multicolumn{1}{c|}{$2$}    & \multicolumn{1}{c|}{$(0,2,1)$}              & \multicolumn{1}{c|}{$21$}          & \multicolumn{1}{c|}{$a_2\oplus b_1$}   & \multicolumn{1}{c|}{$20$}          & \multicolumn{1}{c|}{$a_2$}       & \multicolumn{1}{c|}{$22$}          & $a_2\oplus b_2$   \\ \hline
        \multicolumn{1}{|c|}{$p_{(2)}^{2,[0,1,2]}$}    & \multicolumn{1}{c|}{$2$}    & \multicolumn{1}{c|}{$(0,1,2)$}              & \multicolumn{1}{c|}{$21$}          & \multicolumn{1}{c|}{$a_2\oplus b_1$}   & \multicolumn{1}{c|}{$22$}          & \multicolumn{1}{c|}{$a_2\oplus b_2$}   & \multicolumn{1}{c|}{$20$}          & $a_2$       \\ \hline
        \end{tabular}}
    \end{subtable}
\end{table*}

\subsection{The Reduced W-PIR$^{\#}$ Code}\label{section-reduced-gTSC}

A simpler scheme can in fact be as good as the general W-PIR$^\#$ code in some cases, and this reduced version plays an instrumental role for us to establish the optimal probability allocation for W-PIR$^{\#}$. 
In this reduced version, we set the probability as follows.
\begin{align}
    \Pb_{k}(F^*) = \begin{cases} p_{\#}, &  F^* \in [1:N]\\
    p_{j}, & \begin{array}{c} F^* = (F,\pi)\in [0:N-1]^{K-1}\times \cP \\ \text{      : $\pi$ is cyclic and $\|F\|=j$}  \end{array}\\
    0, & \text{otherwise}
    \end{cases},\label{eqn:reduced}
\end{align}
where $\|F\|$ is the Hamming weight of the first part of the random key $(F,\pi)$ when $F^*\notin [1:N]$. In other words, only cyclic permutations are allowed, instead of the full set of permutations; moreover, $F$'s with the same Hamming weight are assigned the same probability. We say $\pi$ is \emph{cyclic} when $\pi$ is in the set of $\{ \pi: \pi(n+1)= \left( \pi(n)+1 \right)_N, \forall n\in [1:N] \}$. Note that this reduced W-PIR$^{\#}$ is \emph{symmetric} even when used in the heterogeneous setting. 

The query $Q_n^{[k]}$ can take any possible values in $\cQ$. Denote $t_j\triangleq |\{q \in \mathcal{Q}:\|q\|=j\}|$, which is calculated as 
\begin{align}
    t_j= \binom{K}{j}(N-1)^j, ~\forall j\in[0:K].   \label{coefficient-t-j}
\end{align}
For notational simplicity, let $p_{-1}=p_K=0$. Similarly, we use $s_j$ to denote $|\cF_j|$, that is, the number of random keys $f$ that have Hamming weight $j$, given by
\begin{align}
    s_j= \binom{K-1}{j}(N-1)^j, ~\forall j\in[0:K-1].    \label{coefficient-s-j}
\end{align}
 
The download cost and the corresponding privacy of the reduced W-PIR$^{\#}$ code under the Max-L and MI metric are discussed in Section \ref{sec:maxL} and Section \ref{sec:MI}, respectively.

\section{WPIR: The Maximal Leakage Setting} 
\label{sec:maxL}

\subsection{Main Result}

We summarize the main result with \emph{heterogeneous} trustfulness of the servers under the Max-L metric in Theorem \ref{theorem:allocations}. 

\begin{theorem}
\label{theorem:allocations}
An optimal probability assignment for W-PIR$^{\#}$ under the Max-L metric is given by
\begin{align*}
&p_{(\#)}^{k,1}=\frac{N^{K}(1-D+D/N)-1}{N^{K-1}-1}:=\hat{p}_{\#}, \quad k\in [1:K];\\
&p_{(f)}^{k,\pi^*}=\frac{1-\hat{p}_{\#}}{N^{K-1}},\quad k\in [1:K],\quad f\in \Fc,
\end{align*}
where $\pi^*$ is the mapping $\pi^*(n)=n+1$, and other $p_{(\#)}^{k,n}$ and $p_{(f)}^{k,\pi}$ are assigned value zero. Without loss of generality, let $\gamma_1 \leq \gamma_2 \leq \cdots \leq \gamma_n$. As a consequence, with download cost $D\in[1,D^*]$, we have the optimal surrogate leakage for the W-PIR$^\#$ code as 
\begin{align}
\rho^{*}_{(\text{Max-L})}(D)= \gamma_{1} & \frac{(K-1)\left[N^{K-1}(N-(N-1)D)-1 \right] }{N^{K-1}-1} +\sum^{N}_{n=1} \gamma_{n}.
\end{align}
\end{theorem}

Theorem \ref{theorem:allocations} implies that without loss of optimality for the W-PIR$^\#$ code,  we can directly use probabilistic sharing between a direct download from the most trustworthy server and the original TSC strategy without any permutation. In other words, it consists of a completely public part (to the most trusted server) and a completely private part, and the proportion of the mixture determines the exact leakage in this trade-off. Intuitively, this strategy makes perfect sense, since the most trusted server will induce the least amount of leakage, and we might as well prefer to retrieve the whole message from it. Note that the probability assignment given in Theorem \ref{theorem:allocations} for the heterogeneous W-PIR$^\#$ is also an optimal probability assignment for the homogeneous setting under this metric.

The proof of Theorem \ref{theorem:allocations} is however quite sophisticated: first, we establish that without the loss of optimality, we can restrict our attention to a special type of probability allocation strategy corresponding to the reduced W-PIR$^\#$ code for the \textit{homogeneous} setting; then we show that a particular probability allocation for the reduced W-PIR$^\#$ code is in fact optimal for the \textit{homogeneous} setting; lastly, we make a reduction based on a special property in the reduced W-PIR$^\#$ code, to yield the optimal probability allocation for the \textit{heterogeneous} setting. 

The download cost and maximal leakage of the reduced W-PIR$^{\#}$ code are given in the following proposition. The proof is relatively straightforward, and we omit it here for brevity.

\begin{prop}\label{prop:DL}
    The reduced W-PIR$^\#$ scheme induces the download cost and Max-L pair $\left(\rho_{(\text{Max-L})}, D\right)$ given by 
    \begin{align}
        D &=\frac{N-(Np_\#+Np_0)}{N-1},  \\
        \rho_{(\text{Max-L})} &=\sum_{n=1}^N \gamma_n2^{\Lc(M \rightarrow Q^{[M]}_n)}\notag \\
        &=\sum_{n=1}^N{\gamma}_{n}\bigg{(}\sum_{j=1}^K t_j\max\{p_{j-1},p_j\}+p_0 +(N+K-1)p_{\#}\bigg{)},    \label{reduced-gTSC-max-leakage} 
    \end{align}
    for $p_\#\in\left[0,1/N\right]$. 
\end{prop}

\subsection{Homogeneous Trustfulness: Reduced W-PIR$^\#$ is Optimal}

Let us consider the homogeneous case under the Max-L metric where $\gamma_1=\gamma_2=\ldots=\gamma_N=\gamma$, which we shall refer to as problem $P1$: 
\begin{mini}[2]
 {p^{k,n}_{(\#)}, p^{k,\pi}_{(f)}}{\rho_{(\text{Max-L})}=\sum_{n=1}^N \gamma_{n}2^{\Lc(M \rightarrow Q^{[M]}_n)}}{}{}
  \addConstraint{p^{k,n}_{(\#)}}{\geq 0,\quad\forall k,n}{}
  \addConstraint{p^{k,\pi}_{(f)}}{\geq 0,\quad\forall k,\pi,f}{}
  \addConstraint{\sum^{N}_{n=1} p^{k,n}_{(\#)}+\sum_{f}\sum_{\pi}p^{k,\pi}_{(f)}}{=1,\quad\forall k}{}
  \addConstraint{p^k_d + \frac{N}{N-1}(1-p^k_d)}{\leq D,\quad\forall k.}{}
\end{mini}
Recall that $p^{k,n}_{(\#)}$ is the probability of requesting the server-$n$ only for the entire $k^{th}$ message, and $p^{k,\pi}_{(f)}$ is the probability of querying for the $k^{th}$ message under the random key $(f,\pi)$. For simplicity, we will write $p^{k,\pi}_{(\underline{0_{K-1}})}$ as $p^{k,\pi}_{(0)}$ in the sequel.

We first show that the optimal value of the optimization problem $({P1})$ above, which is achieved under the optimal probability distribution in W-PIR$^\#$ code, is the same as the optimal value of the optimization problem $({P2})$ below, which is achieved by the optimal probability allocation for the reduced W-PIR$^\#$:
\begin{mini}[2]
  {p_\#,p_0, p_1,\ldots,p_{K-1}}{\begin{array}{l}\sum_{j=1}^K t_j\max\{p_{j-1},p_j\}\\+(p_0+(N-1)p_{\#})+Kp_{\#}\end{array}}{}{}
  \addConstraint{p_\#,p_0, p_1,\ldots,p_{K-1} }{\geq 0}{}
  \addConstraint{N p_{\#}+\sum_{j=0}^{K-1}Ns_j p_j}{=1}{}
  \addConstraint{\frac{N-(Np_\#+Np_0)}{N-1}}{\leq D.}{}
\end{mini}

For simplicity of notation, we have taken the convention that $p_{-1}=p_K=0$. The following proposition establishes the optimality of the reduced W-PIR$^\#$ code.
\begin{prop}\label{prop:P1P2}
In the homogeneous setting, $(P1)=(P2)$. 
\end{prop}
The proof is given in Appendix \ref{sec:proveProp2} by carefully constructing a sequence of inequalities based mainly on the convexity of the maximum function. The optimality of the probabilistic sharing solution of the reduced W-PIR$^\#$ code in the homogeneous setting is established in Theorem \ref{theorem:homo}.

\begin{theorem}
\label{theorem:homo}
With download cost $D\in [1,D^*]$ and homogeneous trustfulness $\gamma$, the optimal surrogate leakage under the Max-L metric is given as
\begin{align}
&\hat{\rho}_{(\text{Max-L})}^*(D) =\gamma \sum_{n=1}^N 2^{\Lc(M \rightarrow Q^{[M]}_n)}\notag\\
&=N\gamma \left( 1+\frac{(K-1)\left[ N^{K-1}(N-(N-1)D)-1\right]  }{N^{K}-N} \right),
\end{align}
which is achieved using the allocation in Theorem \ref{theorem:allocations}.
\end{theorem}

This theorem is proved in Appendix \ref{sec:proveTh2} by analyzing the KKT conditions \cite{BoydBook} of the given convex optimization problem and constructing primal and dual solutions satisfying these conditions.

\begin{figure}[tb]
    \centering
    \includegraphics[width=\linewidth]{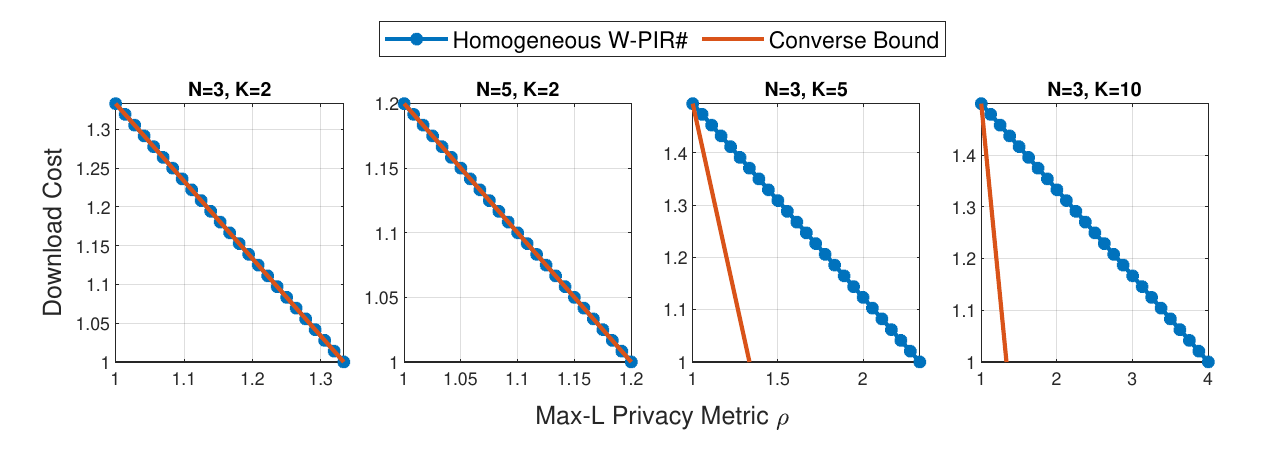}
    \caption{W-PIR$^{\#}$ v.s. Converse Bound in \cite{Lin2021} for the Max-L metric.}
    \label{fig:converse}
\end{figure}

\begin{remark}
(Connections to the reduced W-PIR$^{\#}$ in \cite{qian2022improved}). We first remark that Theorem \ref{theorem:homo} is a special case of the heterogeneous setting, when we set $\gamma_1=\ldots=\gamma_n=\gamma$. The optimal probability allocation is achieved using the same strategy as W-PIR$^{\#}$. Without loss of generality, when $\gamma=1/N$, the optimal pair $(\rho, D)$ is equivalent to the reduced W-PIR$^{\#}$ proposed in \cite{qian2022improved}. However, if $\gamma\neq 1/N$, there exists a multiplicative ratio between $\hat{\rho}_{(\text{Max-L})}^*$ and the optimal $\rho$ in the reduced W-PIR$^{\#}$.
\end{remark}

\begin{remark}
(Connections to the converse bound in \cite{Lin2021}). Consider an $(N,K)$ W-PIR scheme with download cost $D\in [1,D^*]$. Under the assumption that the query size does not grow unbounded and the downloaded answer size per server is lower than or equal to the entire retrieved file size, the converse result for the Max-L metric in \cite{Lin2021} can be written as
\begin{align}
    \rho_{(\text{Max-L})} \geq 1 + \frac{1}{N} - \frac{(D-1)(1-N^{-1})}{1-N^{-(K-1)}}.
\end{align}
The optimal $\hat{\rho}_{(\text{Max-L})}^*$ in Theorem \ref{theorem:homo} in fact matches this outer bound when $K=2$ for any number of servers\footnote{We are grateful to one of the reviewers for bringing this fact to our attention.}. However, for $K>2$, there exists a discrepancy between W-PIR$^{\#}$ and the converse result, which is shown in Fig. \ref{fig:converse}.
\end{remark}

\subsection{Heterogeneous Trustfulness: Proof of Theorem \ref{theorem:allocations}}
\label{sec:proveTh1}

We are now ready to prove Theorem \ref{theorem:allocations}. 

\begin{proof}
 Recall that the loss function, i.e., the objective function, in the \emph{heterogeneous} Max-L setting is 
\begin{align}
\rho_{(\text{Max-L})}= \sum_{n=1}^N \gamma_n2^{\Lc(M \rightarrow Q^{[M]}_n)}.
\end{align}
We can alternatively consider an equivalent loss function $\rho_o$ defined as
\begin{align}
\rho_o= \sum_{n=1}^N \gamma_n\left(2^{\Lc(M \rightarrow Q^{[M]}_n)}-1\right).
\end{align}
We shall denote the optimal value under download cost constraint $D$ as $\rho^*_{(\text{Max-L})}(D)$ for the loss function $\rho_{(\text{Max-L})}$, and similarly for other loss functions in the sequel. 
It is clear that the optimal value ${\rho}^*_{(\text{Max-L})}(D)$ and the optimal value $\rho_o^*(D)$ are related as 
\begin{align}
\rho^*_o(D)= \rho^*_{(\text{Max-L})}(D)-\sum_{n=1}^N \gamma_n.
\end{align}

Next consider a \emph{homogeneous} setting, with the same download cost constraint $D$ and the corresponding surrogate loss function 
\begin{align}
\hat{\rho}_{(\text{Max-L})}= \gamma_1\sum_{n=1}^N 2^{\Lc(M \rightarrow Q^{[M]}_n)},
\end{align}
as well as the corresponding equivalent loss function $\hat{\rho}_o$
\begin{align}
\hat{\rho}_o= \gamma_1\sum_{n=1}^N (2^{\Lc(M \rightarrow Q^{[M]}_n)}-1).
\end{align}
In a similar manner, the optimal $\hat{\rho}^*_{(\text{Max-L})}(D)$ and the optimal $\hat{\rho}_o^*(D)$ are related as
\begin{align}
\hat{\rho}^*_o(D)= \hat{\rho}^*_{(\text{Max-L})}(D)- N\gamma_1.
\end{align}

It is clear that the optimal value of the homogeneous setting $\hat{\rho}^*_o(D)$ is less than or equal to the optimal value of the heterogeneous setting $\rho^*_o(D)$, i.e. $\hat{\rho}^*_o(D)\leq \rho^*_o(D)$, because $\gamma_1\leq \gamma_2\leq\ldots\leq\gamma_N$ and $2^{\Lc(M \rightarrow Q^{[M]}_n)}\geq1$ for any $n$ due to the non-negativity of the maximal leakage metric. Since under this new surrogate loss function the problem is homogeneous, Theorem \ref{theorem:homo} implies that 
\begin{align}
&\hat{\rho}_o^*(D)=\hat{\rho}^*_{(\text{Max-L})}(D)-N\gamma_1\notag\\
=&N\gamma_1 \left(\frac{(K-1)\left[ N^{K-1}(N-(N-1)D)-1\right]  }{N^{K}-N}\right),
\end{align}
which is therefore a lower bound for $\rho^*_o(D)$. It follows that
\begin{align}
&\rho^*_{(\text{Max-L})}(D)=\rho_o^*(D)+\sum_{n=1}^N \gamma_n\geq \hat{\rho}_o^*(D)+\sum_{n=1}^N \gamma_n\notag\\
=&\gamma_1 (K-1)\frac{N^{K-1}(N-(N-1)D)-1 }{N^{K-1}-1}+\sum_{n=1}^N \gamma_n.
\end{align}
However, this lower bound is indeed achieved by the probability distribution assignment in Theorem \ref{theorem:allocations} by assigning 
\begin{align}
\hat{p}_{\#} = \frac{N^K(1-D+D/N)-1}{N^{K-1}-1}.
\end{align}
The proof is thus complete.
\end{proof}

\subsection{Heterogeneous Trustfulness: Numerical Results}
The numerical results shown in Fig. \ref{fig:MaxL} compare W-PIR$^\#$ with TSC (without $p_{\#})$, reduced W-PIR$^\#$ and numerically optimized W-PIR$^\#$. Our newly proposed code (referred to as W-PIR$^\#$) establishes a new benchmark point of minimum download cost and privacy leakage trade-offs, which is validated by numerically solving W-PIR$^\#$ with convex programming tools. The numerical results in both homogeneous and heterogeneous settings corroborate the result established in Theorem \ref{theorem:allocations}: the optimal $(\rho, D)$ trade-offs under the Max-L metric can be achieved by employing a probabilistic sharing of direct download from the most trustworthy server and the original TSC strategy without any permutation. The balance between the fully public segment to the most trustworthy server and the completely private segment is crucial, as it precisely dictates the level of privacy leakage involved under this trade-off. This is logical, as using the most trusted server minimizes leakage, suggesting that retrieving the entire message from this server is a viable strategy. We further note that the probability allocation of W-PIR$^{\#}$ in the heterogeneous setting is equally optimal in the homogeneous setting. Although the reduced W-PIR$^{\#}$ is as good as the W-PIR$^{\#}$ in the homogeneous setting, it is no longer optimal in the heterogeneous setting. As the number of servers and messages increases, W-PIR$^{\#}$ achieves the most favorable $(\rho, D)$ trade-offs compared to other schemes with an increasing divergence.

\begin{figure*}[ht!]
    \centering
    \includegraphics[width=\linewidth]{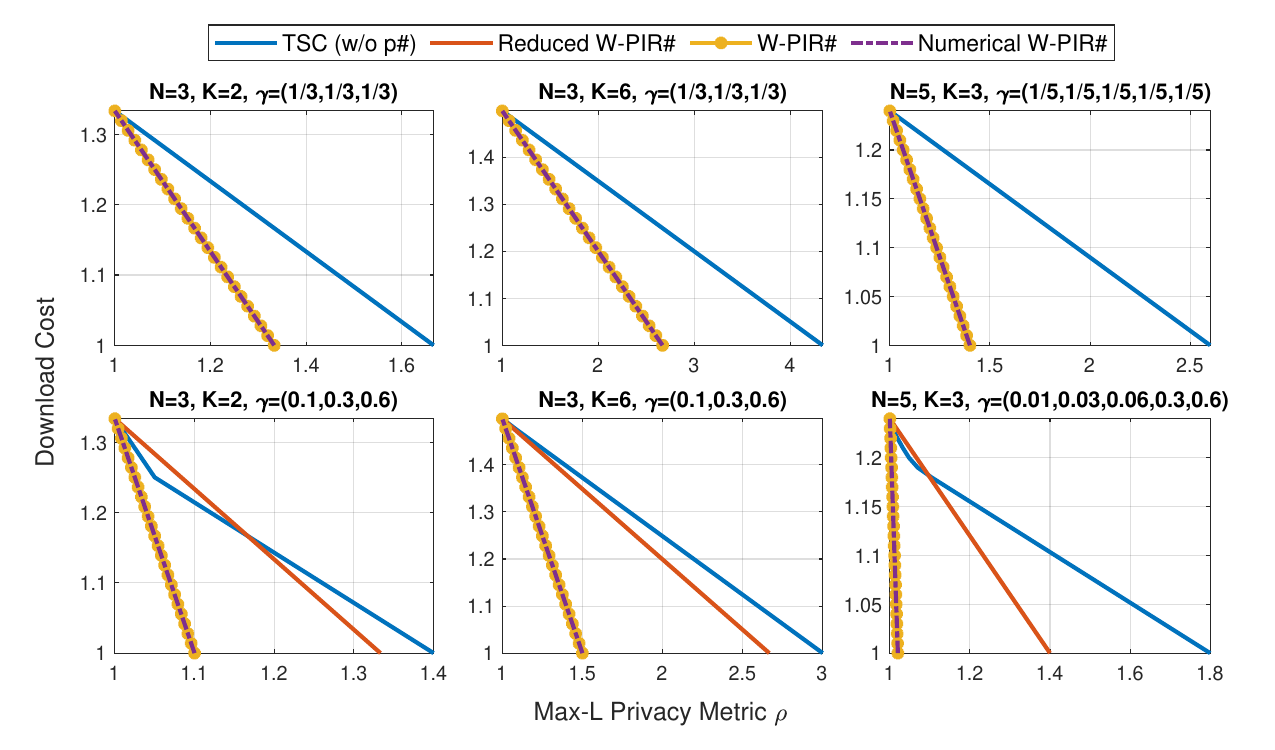}   
    \caption{Numerical comparisons between the proposed code and existing codes under the Max-L metric.}
    \label{fig:MaxL}
\end{figure*}

\vspace{-0.2cm}
\section{WPIR: The Mutual Information Setting} 
\label{sec:MI}

\subsection{Main Result}

Under the MI metric, the analytical probability allocation for W-PIR$^{\#}$ in a \emph{heterogeneous} setting is quite sophisticated. Therefore, we first present the main result with \emph{homogeneous} trustfulness of the servers under the MI metric in Theorem \ref{theorem-MI-mainresult}. 

\begin{theorem}
\label{theorem-MI-mainresult}
An optimal probability allocation for homogeneous W-PIR$^{\#}$ under the MI metric is given by
\begin{align*}
    &p^{k,1}_{(\#) }=Np_{\#}=N\left( 1-p_{0}\right) -\left( N-1\right)D, \quad k\in [1:K],\\
    &p^{k,\pi^{\ast} }_{(0)}=p_0=\left( N-1\right)  (D-1)\left( \sum^{K-1}_{j=1} Ns_{j}\prod^{j}_{i=1} \frac{1}{x_{i}} \right)^{-1},\\
    &p^{k,\pi^{\ast} }_{(f)}=p_j=p_{0}\prod^{j}_{i=1} \frac{1}{x_{i}},\quad
    (f,\pi )\in \left\{ \left( f,\pi^{\ast } \right) : \| f\| =j>0\right\},
\end{align*}
where $\pi^* \in\{ \pi: \pi(n+1)= \left( \pi(n)+1 \right)_N, \forall n\}$, i.e., $\pi^*$ is cyclic, and other $p^{k,n}_{(\#)}$ and $p^{k,\pi}_{(f)}$ are assigned value zero. $\boldsymbol{x}=(x_{1},x_{2},\ldots ,x_{K-1})$ are defined by the following sequence with $x_{1}=\frac{K-1}{K^{\frac{N-2}{N-1} }-1}$:
\begin{align}
    &\log \frac{(K-j)x_{K-j}+j}{K}=\sum^{j-1}_{i=0} (1-N)^{i}\log \frac{(K-1)x_{K-1}+1}{K} \notag\\
    &\qquad\qquad\qquad\qquad\qquad\qquad-\sum^{j-1}_{i=1} (1-N)^{i}\log x_{K-j+i}. \label{theorem-MI-xj}
\end{align}
As a consequence, with download cost $D\in \left[ 1,D^{\ast }\right]$, the optimal surrogate leakage for homogeneous W-PIR$^{\#}$ under the MI metric is 
\begin{align}
  &  \rho^{\ast }_{(\text{MI})} (D) = I(\boldsymbol{p})  := p_{\# }\log K + \frac{1}{K} \sum^{K}_{j=1} \begin{pmatrix}K\\ j\end{pmatrix} (N-1)^{j}\bigg{\{} jp_{j-1}\log p_{j-1} +(K-j)p_{j}\log p_{j} \notag \\
    &\qquad\qquad\qquad\qquad\qquad\qquad-\left[ jp_{j-1}+(K-j)p_{j}\right]\log \frac{jp_{j-1}+(K-j)p_{j}}{K}\bigg{\}}.
\end{align}
\end{theorem}

Theorem \ref{theorem-MI-mainresult} implies that under the MI metric, we can achieve the optimal $\mathcal{G}_{\text{MI}}$ for \emph{homogeneous} W-PIR$^{\#}$ using the similar reduction in the Max-L setting: a public part (to the most trusted server) and a private part. Note that in homogeneous cases, our proposed probability allocation is equivalent to the reduced symmetric code proposed in \cite{qian2022improved}, where only cyclic permutations are required instead of the entire set of permutations. This allocation strategy has the property of time sharing between the public and the private parts, as shown in Proposition \ref{prop:MItimesharing}. We also compared $\rho_{(\text{MI} )}^{\ast}$ with the outer bound in \cite{Lin2021}, which is shown in Fig. \ref{fig:converse-MI}.

\begin{figure}[tb]
    \centering
    \includegraphics[width=0.9\linewidth]{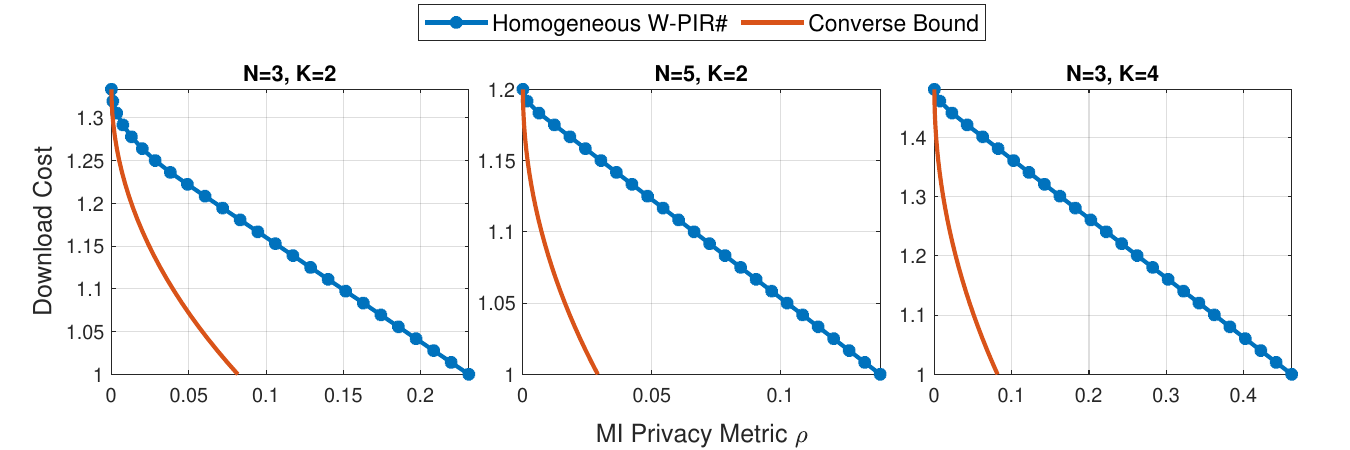}
    \caption{W-PIR$^{\#}$ v.s. Converse Bound in \cite{Lin2021} for the MI metric.}
    \label{fig:converse-MI}
\end{figure}

The construction of probability allocation under the MI metric given above is much more sophisticated than that under the Max-L metric. The auxiliary sequence $\boldsymbol{x}$ is well defined by the above recursive equations when $x_1$ is specified. It is evident that $\boldsymbol{p}=(p_{\#},p_0,p_1,\ldots,p_{K-1})$ generates a probability distribution given to the relevant combinatorial coefficients.

We shall show later that the optimality of this structure is broken down in the \emph{heterogeneous} cases, where a \emph{closed-form} solution becomes intractable. We also present numerical results to illustrate the discrepancy between our proposed homo W-PIR$^{\#}$ allocation and the numerically optimized region $\mathcal{G}_{\text{MI}}$ in a \emph{heterogeneous} setting.

To prove Theorem \ref{theorem-MI-mainresult}, we first establish the reduced W-PIR$^{\#}$ with a special probability assignment strategy in Proposition \ref{prop:P3P4}. We then show that such a reduced W-PIR$^{\#}$ strategy is, in fact, optimal by carefully investigating the properties of MI and the log sum inequality. Lastly, we construct an explicit probability allocation by carefully solving the KKT conditions of the reduced W-PIR$^{\#}$ optimization problem. We defer the complete proof of Theorem \ref{theorem-MI-mainresult} to Section \ref{sec:proof-MI-mainresult}.

\begin{prop}
\label{prop:MItimesharing}
    Denote the achievable region of the clean TSC scheme under the MI metric as $\hat{\Gc}_{\text{MI}}$. In the homogeneous case, the optimal W-PIR$^{\#}$ strategy under the MI metric is time-sharing between the clean TSC scheme and the extreme point using the direct retrieval pattern $\#$ only, i.e., $\Gc_{\text{MI}} = conv\left( \hat{\Gc}_{\text{MI}}\cup \left\{ \left( \gamma_{1}\log K,1\right)\right\}  \right)$.
\end{prop}

\begin{proof}
One direction $conv\left( \hat{\Gc}_{\text{MI}}\cup \left\{ \left( \gamma_{1} \log K,1\right)\right\} \right) \subseteq \Gc_{\text{MI}}$ is trivial because the time-sharing region can be viewed as $\Gc_{\text{MI}}$ under the additional constraints. In the other direction, we need to show that with the same download cost $D$, the time-sharing strategy provides a lower bound of MI leakage $\rho_{(\text{MI})}$. Without loss of generality, we set $\gamma_{1}=\ldots=\gamma_{N}=\gamma=\frac{1}{N}$ in the homogeneous case. Given the optimal surrogate leakage for homogeneous W-PIR$^{\#}$ shown by Theorem \ref{theorem-MI-mainresult}, 
\begin{align}
    \rho_{(\text{MI})} = & \sum^{N}_{n=1} \gamma_{n} \text{MI} \left( M\rightarrow Q^{[M]}_{n}\right) = \sum^{N}_{n=1} \frac{1}{N} I(Q_n;M) \\
    \geq &p_{\# }\log K + \frac{1}{K} \sum^{K}_{j=0} \begin{pmatrix}K\\ j\end{pmatrix} (N-1)^{j} \bigg{\{} jp_{j-1}\log p_{j-1} +(K-j)p_{j}\log p_{j} \notag\\
    &\qquad\qquad\qquad\qquad-\left[ jp_{j-1}+(K-j)p_{j}\right]\log \frac{jp_{j-1}+(K-j)p_{j}}{K}\bigg{\}} \\
    = & \alpha \frac{\log K}{N} + (1-\alpha) \frac{1}{K} \sum^{K}_{j=0} \begin{pmatrix}K\\ j\end{pmatrix} (N-1)^{j} \bigg{\{} jp'_{j-1}\log p'_{j-1} +(K-j)p'_{j}\log p'_{j} \notag\\
    &\qquad\qquad\qquad\qquad-\left[ jp'_{j-1}+(K-j)p'_{j}\right]\log \frac{jp'_{j-1}+(K-j)p'_{j}}{K}\bigg{\}},
\end{align}
where 
\begin{align}
    \alpha &= Np_{\#}, \\
    1-\alpha &= \sum^{K-1}_{j=0} Ns_{j}p_{j}, \\
    p^{\prime}_{j} &= \frac{p_j}{1-\alpha}.
\end{align}

The first term $\log(K)/N$ is achieved by the extreme point using escape pattern $\#$ only to retrieve the whole message from the most trustworthy server. For the second term, note that the probability vector $\boldsymbol{p}'=\left( p'_0, p'_1, \ldots, p'_{K-1}\right)$ induces a valid probability space since 
\begin{align}
    \sum^{K-1}_{j=0} Ns_{j}p'_{j} = \frac{\sum^{K-1}_{j=0} Ns_{j}p_{j} }{1-\alpha} =1.
\end{align}

The optimal probability allocation of clean TSC scheme without the escape retrieval symbol $\#$ can be induced from W-PIR$^{\#}$ by setting $p^{k,1}_{(\#)}$ as zero \cite{qian2022improved}:
\begin{align}
    \hat{I}(\boldsymbol{p}) =& \frac{1}{K} \sum^{K}_{j=0} \begin{pmatrix}K\\ j\end{pmatrix} (N-1)^{j}\bigg{\{} jp_{j-1}\log p_{j-1} +(K-j)p_{j}\log p_{j}\notag\\
    &\qquad\qquad\qquad\qquad-\left[ jp_{j-1}+(K-j)p_{j}\right]\log \frac{jp_{j-1}+(K-j)p_{j}}{K}\bigg{\}},
\end{align}
where $\boldsymbol{p}=\left( p_0, p_1, \ldots, p_{K-1}\right)$ is obtained by using the same allocation strategy in Theorem \ref{theorem-MI-mainresult} but setting $p^{k,1}_{(\#)}$ as zero. The rigorous proof of $\hat{I}(\boldsymbol{p})$ can be found in Appendix \ref{proof:P3P4}.

Therefore, we have the other direction $conv\left( \hat{\Gc}_{\text{MI}}\cup \left\{ \left( \gamma_{1} \log K,1\right)\right\} \right) \supseteq \Gc_{\text{MI}}$, i.e., the optimal W-PIR$^{\#}$ strategy under the MI metric is time-sharing between the extreme point (using the direct retrieval $p^{k,1}_{(\#)}$ only) and the clean TSC scheme (using $p^{k,\pi}_{(f)}$ only).
\end{proof}

Note that this equivalence to the time-sharing (probabilistic sharing) solution is based on one particular structure of the optimal solution in the homogeneous setting. More precisely, it essentially relies on the fact that in the optimal solution, the ``request-for-nothing" query does not induce any privacy leakage.

\subsection{Homogeneous Trustfulness: Reduced W-PIR$^\#$ is Optimal}
Given the parameters of the server trustworthy $\gamma_{n}, n\in [1:N]$, the heterogeneous W-PIR$^\#$ objective under the MI metric is shown in the following optimization problem $(P3)$:

\begin{mini}[2]
{p^{k,n}_{(\#)}, p^{k,\pi}_{(f)}}{\rho_{(\text{MI})}=\sum^{N}_{n=1} \gamma_{n} \text{MI} \left( M\rightarrow Q^{[M]}_{n}\right) }{}{}
\addConstraint{p^{k,n}_{(\#)}}{\geq 0,\quad\forall k,n}{}
\addConstraint{p^{k,\pi}_{(f)}}{\geq 0,\quad\forall k,\pi,f}{}
\addConstraint{\sum^{N}_{n=1} p^{k,n}_{(\#)}+\sum_{f}\sum_{\pi}p^{k,\pi}_{(f)}}{=1,\quad\forall k}{}
\addConstraint{p^k_d + \frac{N}{N-1}(1-p^k_d)}{\leq D,\quad\forall k.}{}
\end{mini}

\begin{prop}
\label{prop:P3P4}
The above optimization problem $({P3})$ in the homogeneous setting, where $\gamma_{1}=\gamma_{2}=\ldots=\gamma_{N}=\gamma=1/N$, has the same optimal value as the following reduced problem $({P4})$:
\begin{mini}[2]
{p_{\# },p_{0},\ldots,p_{K-1}}{I(\boldsymbol{p})\qquad\qquad\qquad\qquad\qquad}{}{}
\addConstraint{p_{\# },p_{0},p_{1},\ldots ,p_{K-1}}{\geq 0}{}
\addConstraint{Np_{\# }+\sum^{K-1}_{j=0} Ns_{j}p_{j}}{=1}{}
\addConstraint{\frac{N-N(p_{\# }+p_{0})}{N-1} }{\leq D.}{}
\end{mini}
\end{prop}

We prove this proposition in Appendix \ref{proof:P3P4} mainly based on the properties of mutual information and log sum inequality. The optimal probability assignment in Proposition \ref{prop:P3P4} is established by the scheme in Theorem \ref{theorem-MI-mainresult}, which is proved in Section \ref{sec:proof-MI-mainresult}.

Before the proof of Theorem \ref{theorem-MI-mainresult}, we further specify that the download cost constraint $p_{0}+p_{\# }\geq 1-D+D/N$ under the homogeneous MI setting can always be achieved by letting $p_{0}+p_{\# }=1-D+D/N$ in Lemma \ref{lemma-MI-achieve-download}.

\begin{lemma}
\label{lemma-MI-achieve-download}
Given a legal solution $\boldsymbol{p}=(p_{\#},p_{0},\ldots,p_{K-1})$ of $(P4)$ satisfying $p_{0}+p_{\# }=\hat{p}^{\ast} \geq \hat{p}=1-D+D/N$, we can explicitly construct a new vector $\boldsymbol{p}'=(p^{\prime}_{\#},p^{\prime}_{0},\ldots,p^{\prime}_{K-1})$ by setting
\begin{align*}
    & p^{\prime }_{j}=\frac{1-N\hat{p} }{1-N\hat{p}^{\ast } } p_{j}, \quad j\in[0,K-1],\\
    & p^{\prime }_{\# }=\hat{p} -p^{\prime }_{0},
\end{align*}
such that this newly constructed assignment has lower MI leakage $I(\boldsymbol{p}') \leq I(\boldsymbol{p})$.
\end{lemma}

The proof of Lemma \ref{lemma-MI-achieve-download} is given in Appendix \ref{proof:lemma-MI-achieve-download} by carefully computing the difference between the new constructed assignment $I(\boldsymbol{p}')$ and the original $I(\boldsymbol{p})$. 

\subsection{Homogeneous Trustfulness: Proof of Theorem \ref{theorem-MI-mainresult}}
\label{sec:proof-MI-mainresult}
Now we are ready to prove Theorem \ref{theorem-MI-mainresult}.

\begin{proof}
Using a similar manner in Max-L, we define $\hat{p} =1-D+D/N$, where $\hat{p} \in \left[ N^{-K},N^{-1}\right]$. Without loss of optimality, a feasible download cost $D\in[1,D^{*}]$ in $(P4)$ can be achieved by setting $p_{0}+p_{\# }=\hat{p}$, which is established in Lemma \ref{lemma-MI-achieve-download}. Therefore, the optimization problem under the homogeneous MI setting can be written as $(P4')$:
\begin{mini}[2]
{p_{\# },p_{0},p_{1},\ldots ,p_{K-1}}{I(\boldsymbol{p})\qquad \qquad \qquad \qquad \qquad}{}{}
\addConstraint{p_{\# },p_{0},p_{1},\ldots ,p_{K-1}}{\geq 0}{}
\addConstraint{Np_{\# }+\sum^{K-1}_{j=0} Ns_{j}p_{j}}{=1}{}
\addConstraint{p_{0}+p_{\#} }{= \hat{p}}{}
\end{mini}

The Lagrangian function is 
\begin{align}
    \mathscr{L}=&I(\boldsymbol{p})-\sum^{K-1}_{j=0} \lambda_{j} p_{j}-\eta_{\# } p_{\# } \notag \\
    &+\nu \left( Np_{\# }+\sum^{K-1}_{j=0} Ns_{j}p_{j}-1\right) \notag \\
    &+\mu \left( \hat{p} -p_{0}-p_{\# }\right)   
\end{align}

We introduce the two sets of auxiliary variables for $j=1,2,\ldots,K-1$:
\begin{align}
    x_{j} &\triangleq p_{j-1}/p_{j}, \\
    y_{j} &\triangleq \log \frac{jx_{j}+K-j}{K}.
\end{align}

Then the KKT condition can be derived as follows:
\begin{enumerate}
    \item stationarity:
    \begin{align}
        \begin{cases}\log K-\eta_{\# } +N\nu -\mu =0\\ \begin{array}{l}s_{j}\bigg{(}-y_{j}+(N-1)\big{(}\log x_{j+1}-y_{j+1}\big{)}\\+N\nu \bigg{)} -\lambda_{j}=0,j\in [1:K-2]\end{array}\\ (N-1)^{K-1}\left[ -y_{k-1}+N\nu \right]  -\lambda_{K-1} =0\\ (N-1)(\log x_{1}-y_{1})+N\nu -\lambda_{0} -\mu =0\end{cases} 
    \end{align}

    \item primal feasibility:
    \begin{align}
        \begin{cases}Np_{\# }+N\sum^{K-1}_{j=0} s_{j}p_{j}-1=0\\ p_{j}\geq 0,p_{\# }\geq 0,j\in [0:K-1]\\ p_{0}+p_{\# }= \hat{p}\end{cases} 
    \end{align}

    \item dual feasibility:
    \begin{align}
        \begin{cases}\eta_{\# } \geq 0\\ \lambda_{j} \geq 0,j\in [0:K-1] \end{cases} 
    \end{align}

    \item complementary slackness:
    \begin{align}
        \begin{cases}\eta_{\# } p_{\# }=0\\ \lambda_{j} p_{j}=0\\ \mu \left( \hat{p} -p_{0}-p_{\# }\right)  =0\end{cases} 
    \end{align}
\end{enumerate}

We give the solution to the KKT conditions as the following variable assignments:
\begin{enumerate}[1)]
\item primal variables:
\begin{align}
    \begin{cases}p_{\# }=\hat{p} -p_{0}\\ p_{0}=(N-1)(D-1)\left( \sum^{K-1}_{j=1} Ns_{j}\prod^{j}_{i=1} \frac{1}{x_{i}} \right)^{-1}  \\ p_{j}=p_{0}\prod^{j}_{i=1} \frac{1}{x_{i}} ,j\in [1:K-1] \end{cases} 
\end{align}
where $\boldsymbol{x}=(x_{1},x_{2},\ldots ,x_{K-1})$ is defined as the following sequence with $x_{1}=\frac{K-1}{K^{\frac{N-2}{N-1} }-1}$:
\begin{align}
    \log \frac{(K-j)x_{K-j}+j}{K} = &\sum^{j-1}_{i=0} (1-N)^{i}y_{K-1} \notag \\
    & -\sum^{j-1}_{i=1} (1-N)^{i}\log x_{K-j+i}.
\end{align}

\item dual variables:
\begin{align}
\begin{cases}\eta_{\# } =0\\ \lambda_{j} =0,j\in [0:K-1]\\ \nu =\frac{y_{K-1}}{N}\\ \mu =\log K+y_{K-1}\end{cases} 
\end{align}
\end{enumerate}

It can be verified that the solution given above satisfies all KKT conditions. Firstly with $\lambda_{j}=0$ for $j=0,1,\ldots,K-1$, $x_{j}$ and $y_{j}$ assigned by (\ref{theorem-MI-xj}), with $y_{j}$'s eliminated, can properly satisfy stationarity and primal feasibility conditions. The dual feasibility and complementary slackness can be easily verified by simply plugging the variables, and this completes the proof.
\end{proof}

\begin{figure*}[tb]
    \centering
    \includegraphics[width=\linewidth]{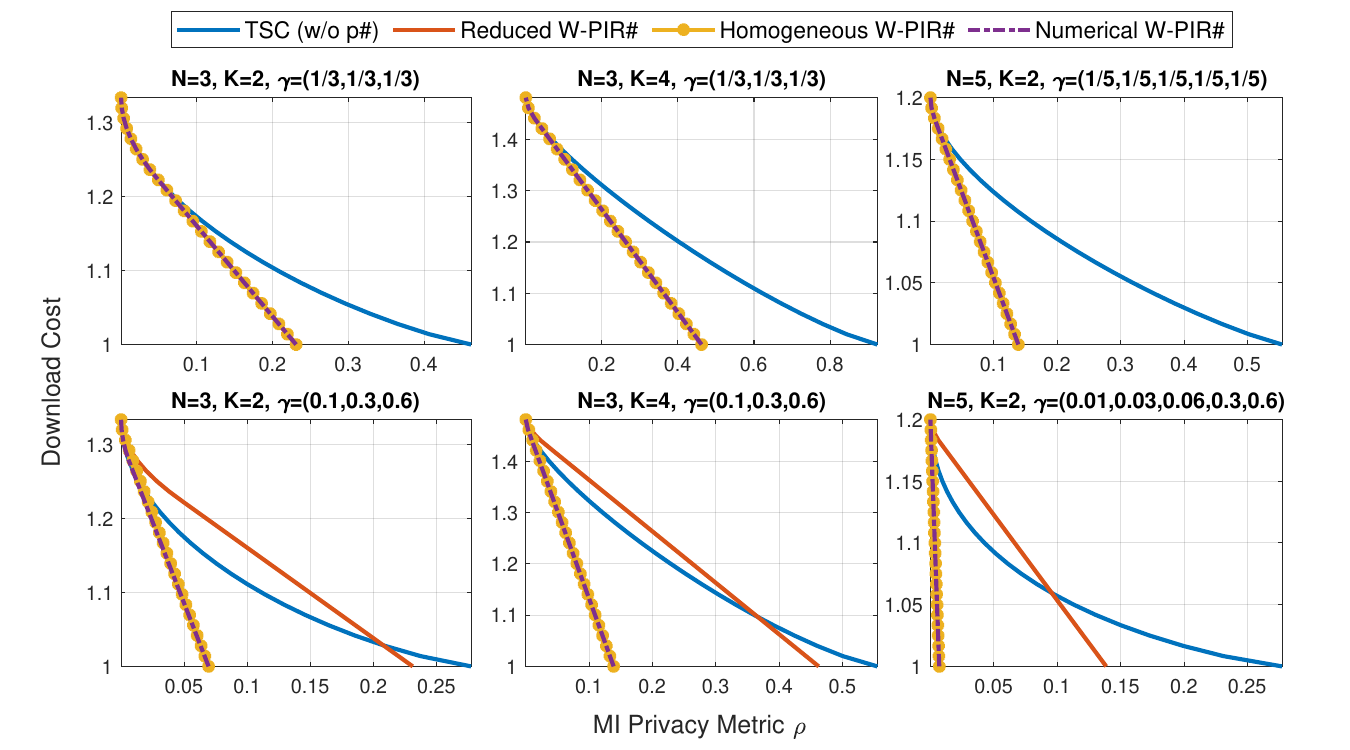}    
    \caption{Numerical comparisons between the proposed code and existing codes under the MI metric.}
    \label{fig:MI}
\end{figure*}

\begin{figure*}[h!]
    \centering
    \begin{subfigure}{.45\textwidth}
        \includegraphics[width=\textwidth]{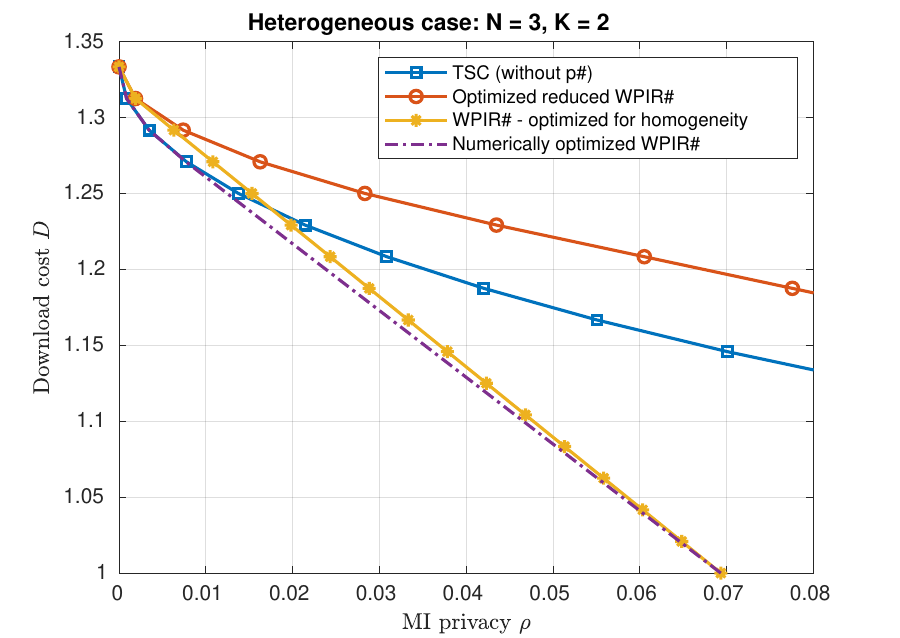}
        \caption{$\gamma = [0.1, 0.3, 0.6]$}
    \end{subfigure}
    \begin{subfigure}{.45\textwidth}
        \includegraphics[width=\textwidth]{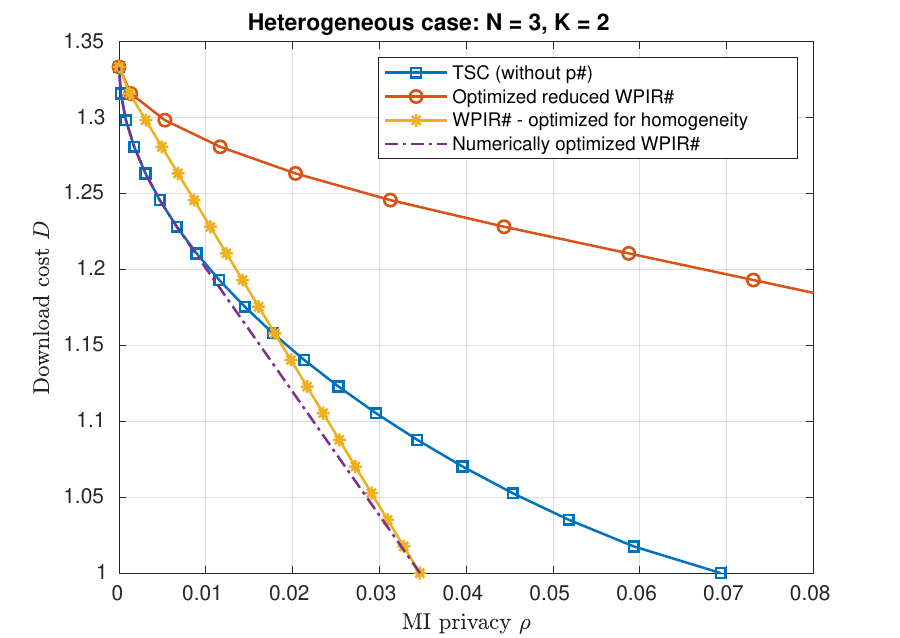}
        \caption{$\gamma = [0.05, 0.05, 0.9]$}
    \end{subfigure}
    \caption{Numerical comparisons with different trustworthy weights $\gamma$ under the MI metric.}
    \label{fig:MIgamma}
\end{figure*}

\subsection{Heterogeneous Trustfulness: Numerical Results}
The optimal regions $\mathcal{G}_{\text{MI}}$ under different allocation strategies are illustrated in Fig. \ref{fig:MI}. In homogeneous settings, our proposed allocation (WPIR$^{\#}$ - optimized for homogeneity) in Theorem \ref{theorem-MI-mainresult} provides the least MI privacy leakage under the same download cost constraint, which indeed matches the result numerically solved by convex programming tools. However, in the heterogeneous case, the WPIR$^{\#}$ optimized for homogeneity is no longer optimal, and a closed-form solution for heterogeneous settings becomes intractable. Therefore, a numerically optimized WPIR$^{\#}$ with convex programming tools is provided to corroborate the theoretical analysis. With an increasing number of servers and messages, the WPIR$^{\#}$ achieves better $(\rho, D)$ trade-offs than other PIR schemes. Although the analytical allocation in Theorem \ref{theorem-MI-mainresult} may not be optimal in the heterogeneous setting, it appears to be quite close to the numerical solution when $N$ and $K$ are large.

We compared the W-PIR schemes with different weights on the trustworthiness of the servers in Fig. \ref{fig:MIgamma}. The allocation of W-PIR$^{\#}$ in Theorem \ref{theorem-MI-mainresult} approaches the numerical solution as the value of $\Delta \gamma \triangleq \gamma_N - \gamma_1$ becomes smaller, that is, there do not exist extremely trustworthy and extremely untrustworthy servers. When $\Delta\gamma \rightarrow 0$, the analytical allocation of W-PIR$^{\#}$ converges to the numerical solution, as shown in homogeneous settings. It is also observed that the original TSC code and the optimized reduced WPIR$^\#$ suffer in a heterogeneous setting because they do not take into account the more trustworthy server.

The structure of the probabilistic sharing mechanism becomes more intricate in the context of MI metric than in the Max-L metric.  In particular, it is observed that when the leakage parameter $\rho$ falls below a certain threshold, these new direct download patterns become ineffective, implying that they are not used in the retrieval process. This numerical result also validates the correctness of Proposition \ref{prop:MItimesharing}: the optimal W-PIR$^{\#}$ is time-sharing between the extreme point of direct retrieval pattern and the clean TSC scheme without the $\#$ pattern for the homogeneous setting, and appears to even hold for the heterogeneous case, though we are not able to establish this conjecture rigorously. However, we can still apply the allocation in Theorem \ref{theorem-MI-mainresult}, which is optimal for the homogeneous setting, which appears to be quite close to the optimal $\mathcal{G}_{\text{MI}}$ numerically computed in this setting, as shown in Fig. \ref{fig:MI}.

\section{Conclusion}
\label{sec:conclusion}
We studied the problem of weakly private information retrieval when there is heterogeneity in the servers' trustfulness, and identified the optimal probability allocation of a general class of W-PIR code, which we refer to as the W-PIR$^\#$ code. This optimal distribution is notably straightforward, essentially being a probabilistic sharing of a capacity-achieving PIR code and a direct download from the most reliable server. Intriguingly, we found that a specific optimal code designed for a homogeneous trust environment is equally effective in a heterogeneous trust scenario, particularly under the Max-L metric. Further explorations in our subsequent work delved into the W-PIR$^\#$ code within the context of the MI metric, examining both homogeneous and heterogeneous settings, where the optimal solutions for the two settings showed significant divergence.

\bibliographystyle{IEEEtran}

\begin{thebibliography}{10}
\providecommand{\url}[1]{#1}
\csname url@samestyle\endcsname
\providecommand{\newblock}{\relax}
\providecommand{\bibinfo}[2]{#2}
\providecommand{\BIBentrySTDinterwordspacing}{\spaceskip=0pt\relax}
\providecommand{\BIBentryALTinterwordstretchfactor}{4}
\providecommand{\BIBentryALTinterwordspacing}{\spaceskip=\fontdimen2\font plus
\BIBentryALTinterwordstretchfactor\fontdimen3\font minus
  \fontdimen4\font\relax}
\providecommand{\BIBforeignlanguage}[2]{{%
\expandafter\ifx\csname l@#1\endcsname\relax
\typeout{** WARNING: IEEEtran.bst: No hyphenation pattern has been}%
\typeout{** loaded for the language `#1'. Using the pattern for}%
\typeout{** the default language instead.}%
\else
\language=\csname l@#1\endcsname
\fi
#2}}
\providecommand{\BIBdecl}{\relax}
\BIBdecl

\bibitem{Chor1995}
B.~Chor, O.~Goldreich, E.~Kushilevitz, and M.~Sudan, ``Private information
  retrieval,'' in \emph{IEEE 36th Annual Foundations of Computer Science},
  Milwaukee, WI, USA, Oct. 1995, pp. 41--50.

\bibitem{Sun2017}
H.~Sun and S.~A. Jafar, ``The capacity of private information retrieval,''
  \emph{IEEE Transactions on Information Theory}, vol.~63, no.~7, pp.
  4075--4088, 2017.

\bibitem{tian2019capacity}
C.~Tian, H.~Sun, and J.~Chen, ``Capacity-achieving private information
  retrieval codes with optimal message size and upload cost,'' \emph{IEEE
  Transactions on Information Theory}, vol.~65, no.~11, pp. 7613--7627, 2019.

\bibitem{t1}
K.~Banawan and S.~Ulukus, ``The capacity of private information retrieval from
  {Byzantine} and colluding databases,'' \emph{IEEE Transactions on Information
  Theory}, vol.~65, no.~2, pp. 1206--1219, Feb. 2019.

\bibitem{t2}
H.~Sun and S.~A. Jafar, ``The capacity of robust private information retrieval
  with colluding databases,'' \emph{IEEE Transactions on Information Theory},
  vol.~64, no.~4, pp. 2361--2370, Apr. 2018.

\bibitem{zhou2022two}
R.~Zhou, C.~Tian, H.~Sun, and J.~S. Plank, ``Two-level private information
  retrieval,'' \emph{IEEE Journal on Selected Areas in Information Theory},
  vol.~3, no.~2, pp. 337--349, 2022.

\bibitem{c1}
R.~Zhou, C.~Tian, H.~Sun, and T.~Liu, ``Capacity-achieving private information
  retrieval codes from {MDS}-coded databases with minimum message size,''
  \emph{IEEE Transactions on Information Theory}, vol.~66, no.~8, pp.
  4904--4916, Aug. 2020.

\bibitem{c2}
T.~Guo, R.~Zhou, and C.~Tian, ``New results on the storage-retrieval tradeoff
  in private information retrieval systems,'' \emph{IEEE Journal on Selected
  Areas in Information Theory}, vol.~2, no.~1, pp. 403--414, Mar. 2021.

\bibitem{c3}
C.~Tian, ``On the storage cost of private information retrieval,'' \emph{IEEE
  Transactions on Information Theory}, vol.~66, no.~12, pp. 7539--7549, Dec.
  2020.

\bibitem{tian2023shannon}
C.~Tian, H.~Sun, and J.~Chen, ``A {Shannon}-theoretic approach to the
  storage--retrieval trade-off in pir systems,'' \emph{Information}, vol.~14,
  no.~1, p.~44, 2023.

\bibitem{c5}
H.~Sun and C.~Tian, ``Breaking the {MDS-PIR} capacity barrier via joint storage
  coding,'' \emph{Information}, vol.~10, no.~9, Aug. 2019.

\bibitem{c6}
K.~Banawan and S.~Ulukus, ``The capacity of private information retrieval from
  coded databases,'' \emph{IEEE Transactions on Information Theory}, vol.~64,
  no.~3, pp. 1945--1956, Mar. 2018.

\bibitem{c7}
R.~Tajeddine, O.~W. Gnilke, and S.~El~Rouayheb, ``Private information retrieval
  from {MDS} coded data in distributed storage systems,'' \emph{IEEE
  Transactions on Information Theory}, vol.~64, no.~11, pp. 7081--7093, Nov.
  2018.

\bibitem{c8}
R.~Freij-Hollanti, O.~W. Gnilke, C.~Hollanti, and D.~A. Karpuk, ``Private
  information retrieval from coded databases with colluding servers,''
  \emph{SIAM Journal on Applied Algebra and Geometry}, vol.~1, no.~1, pp.
  647--664, Nov. 2017.

\bibitem{c9}
H.~Sun and S.~A. Jafar, ``Private information retrieval from {MDS} coded data
  with colluding servers: Settling a conjecture by {Freij-Hollanti} et al.''
  \emph{IEEE Transactions on Information Theory}, vol.~64, no.~2, pp.
  1000--1022, Feb. 2018.

\bibitem{c10}
S.~Kumar, H.-Y. Lin, E.~Rosnes, and A.~Graell~i Amat, ``Achieving maximum
  distance separable private information retrieval capacity with linear
  codes,'' \emph{IEEE Transactions on Information Theory}, vol.~65, no.~7, pp.
  4243--4273, Jul. 2019.

\bibitem{zhu2019new}
J.~Zhu, Q.~Yan, C.~Qi, and X.~Tang, ``A new capacity-achieving private
  information retrieval scheme with (almost) optimal file length for coded
  servers,'' \emph{IEEE Transactions on Information Forensics and Security},
  vol.~15, pp. 1248--1260, 2019.

\bibitem{vardy2023private}
A.~Vardy and E.~Yaakobi, ``Private information retrieval without storage
  overhead: Coding instead of replication,'' \emph{IEEE Journal on Selected
  Areas in Information Theory}, vol.~4, pp. 286--301, Jul. 2023.

\bibitem{d1}
T.~Guo, R.~Zhou, and C.~Tian, ``On the information leakage in private
  information retrieval systems,'' \emph{IEEE Transactions on Information
  Forensics and Security}, vol.~15, pp. 2999--3012, Mar. 2020.

\bibitem{d2}
H.~Sun and S.~A. Jafar, ``The capacity of symmetric private information
  retrieval,'' \emph{IEEE Transactions on Information Theory}, vol.~65, no.~1,
  pp. 322--329, Jan. 2019.

\bibitem{d3}
Z.~Wang, K.~Banawan, and S.~Ulukus, ``Private set intersection: A multi-message
  symmetric private information retrieval perspective,'' \emph{IEEE
  Transactions on Information Theory}, vol.~68, no.~3, pp. 2001--2019, 2021.

\bibitem{s1}
R.~Tandon, ``The capacity of cache aided private information retrieval,'' in
  \emph{2017 55th Annual Allerton Conference on Communication, Control, and
  Computing (Allerton)}, Monticello, IL, USA, Oct. 2017, pp. 1078--1082.

\bibitem{s2}
Y.-P. Wei, K.~Banawan, and S.~Ulukus, ``Fundamental limits of cache-aided
  private information retrieval with unknown and uncoded prefetching,''
  \emph{IEEE Transactions on Information Theory}, vol.~65, no.~5, pp.
  3215--3232, May 2019.

\bibitem{s3}
S.~Kadhe, B.~Garcia, A.~Heidarzadeh, S.~El~Rouayheb, and A.~Sprintson,
  ``Private information retrieval with side information,'' \emph{IEEE
  Transactions on Information Theory}, vol.~66, no.~4, pp. 2032--2043, Apr.
  2020.

\bibitem{s4}
Z.~Chen, Z.~Wang, and S.~A. Jafar, ``The capacity of {$T$-}private information
  retrieval with private side information,'' \emph{IEEE Transactions on
  Information Theory}, vol.~66, no.~8, pp. 4761--4773, Aug. 2020.

\bibitem{s5}
Y.-P. Wei and S.~Ulukus, ``The capacity of private information retrieval with
  private side information under storage constraints,'' \emph{IEEE Transactions
  on Information Theory}, vol.~66, no.~4, pp. 2023--2031, Apr. 2020.

\bibitem{s6}
S.~Li and M.~Gastpar, ``Single-server multi-message private information
  retrieval with side information: the general cases,'' in \emph{2020 IEEE
  International Symposium on Information Theory (ISIT)}, Los Angeles, CA, USA,
  Jun. 2020, pp. 1083--1088.

\bibitem{s7}
Z.~Wang and S.~Ulukus, ``Symmetric private information retrieval with user-side
  common randomness,'' in \emph{2021 IEEE International Symposium on
  Information Theory (ISIT)}, Melbourne, Victoria, Australia, Jul. 2021, pp.
  2119--2124.

\bibitem{lu2023single}
Y.~Lu and S.~A. Jafar, ``On single server private information retrieval with
  private coded side information,'' \emph{IEEE Transactions on Information
  Theory}, vol.~69, no.~5, pp. 3263--3284, Mar. 2023.

\bibitem{ulukus2022private}
S.~Ulukus, S.~Avestimehr, M.~Gastpar, S.~A. Jafar, R.~Tandon, and C.~Tian,
  ``Private retrieval, computing, and learning: Recent progress and future
  challenges,'' \emph{IEEE Journal on Selected Areas in Communications},
  vol.~40, no.~3, pp. 729--748, 2022.

\bibitem{Asonov2002}
D.~Asonov and J.~C. Freytag, ``Repudiative information retrieval,'' in
  \emph{2002 ACM Workshop on Privacy in the Electronic Society}, Washington,
  DC, USA, Nov. 2002, pp. 32--40.

\bibitem{Toledo2016}
R.~R. Toledo, G.~Danezis, and I.~Goldberg, ``Lower-cost $\epsilon$-private
  information retrieval,'' in \emph{2016 Privacy Enhancing Technologies
  Symposium (PETS)}, Darmstadt, Germany, Jul. 2016, pp. 184--201.

\bibitem{Samy2019}
I.~Samy, R.~Tandon, and L.~Lazos, ``On the capacity of leaky private
  information retrieval,'' in \emph{2019 IEEE International Symposium on
  Information Theory (ISIT)}, Paris, France, Jul. 2019, pp. 1262--1266.

\bibitem{ZhuqingJia2019}
Z.~Jia, ``On the capacity of weakly-private information retrieval,'' Master's
  thesis, University of California, Irvine, CA, 2019.

\bibitem{Lin2019}
H.-Y. Lin, S.~Kumar, E.~Rosnes, A.~Graell~i Amat, and E.~Yaakobi,
  ``Weakly-private information retrieval,'' in \emph{2019 IEEE International
  Symposium on Information Theory (ISIT)}, Paris, France, Jun. 2019, pp.
  1257--1261.

\bibitem{Zhou2020a}
R.~Zhou, T.~Guo, and C.~Tian, ``Weakly private information retrieval under the
  maximal leakage metric,'' in \emph{2020 IEEE International Symposium on
  Information Theory (ISIT)}, Los Angeles, CA, USA, Jun. 2020, pp. 1089--1094.

\bibitem{Lin2021}
H.-Y. Lin, S.~Member, S.~Kumar, E.~Rosnes, A.~{Graell i Amat}, and E.~Yaakobi,
  ``Multi-server weakly-private information retrieval,'' \emph{IEEE
  Transactions on Information Theory}, vol.~68, no.~2, pp. 1197--1219, 2022.

\bibitem{Samy2021}
I.~Samy, M.~Attia, R.~Tandon, and L.~Lazos, ``Asymmetric leaky private
  information retrieval,'' \emph{IEEE Transactions on Information Theory},
  vol.~67, no.~8, pp. 5352--5369, Aug. 2021.

\bibitem{lin2021capacity}
H.-Y. Lin, S.~Kumar, E.~Rosnes, A.~Graell~i Amat, and E.~Yaakobi, ``The
  capacity of single-server weakly-private information retrieval,'' \emph{IEEE
  Journal on Selected Areas in Information Theory}, vol.~2, no.~1, pp.
  415--427, 2021.

\bibitem{yakimenka2022optimal}
Y.~Yakimenka, H.-Y. Lin, E.~Rosnes, and J.~Kliewer, ``Optimal
  rate-distortion-leakage tradeoff for single-server information retrieval,''
  \emph{IEEE Journal on Selected Areas in Communications}, vol.~40, no.~3, pp.
  832--846, 2022.

\bibitem{Issa2020}
I.~Issa, A.~B. Wagner, and S.~Kamath, ``An operational approach to information
  leakage,'' \emph{IEEE Transactions on Information Theory}, vol.~66, no.~3,
  pp. 1625--1657, Mar. 2020.

\bibitem{qian2022improved}
C.~Qian, R.~Zhou, C.~Tian, and T.~Liu, ``Improved weakly private information
  retrieval codes,'' in \emph{Proc. 2022 IEEE International Symposium on
  Information Theory (ISIT)}, Jul. 2022, pp. 2827--2832.

\bibitem{BoydBook}
S.~Boyd and L.~Vandenberghe, \emph{Convex Optimization}, 1st~ed.\hskip 1em plus
  0.5em minus 0.4em\relax Cambridge University Press, 2004.

\end{thebibliography}

\appendix

\section{Proof of Proposition \ref{prop:P1P2}}
\label{sec:proveProp2}

\begin{proof}
\noindent\textit{The direction $(P1)\leq(P2)$:} Given the explicit formulas given in Proposition \ref{prop:DL}, this direction is trivially true since $(P2)$ can be viewed as $(P1)$ under the additional constraints enforced through (\ref{eqn:reduced}).

\noindent\textit{The direction $(P1)\geq(P2)$:} Given an optimal solution in $(P1)$, we can find the following assignment of $p_\#,p_0,\ldots,p_{K-1}$:
	\begin{align}
		p_j&=\frac{1}{NKs_j}\sum_{k=1}^{K}\sum_{\pi\in \cP}\sum_{f\in\cF_j} p_{(f)}^{k,\pi}, \quad j\in [0:K-1],   \label{def-p-j-proof}\\
		p_\# &= \frac{1}{NK}\sum_{n=1}^N\sum_{k=1}^K p_{(\#)}^{k,n},\label{def-p-sharp-proof}
	\end{align}
	With the relation (\ref{def-p-j-proof}), we have 
	$p_j \geq 0$ for any $j \in [0:K-1]$, and moreover, 
	\begin{align}
       &\sum_{j=0}^{K-1} Ns_jp_j + Np_\#\ \notag \\
       = &\frac{1}{K}\left[\sum_{k=1}^{K}\left(\sum_{\pi\in \cP}\sum_{f\in\cF}p_{(f)}^{k,\pi} +\sum_{n=1}^N p_{(\#)}^{k,n}\right)\right]=1, \label{prob-constraint-proof}
	\end{align}
	and 
	\begin{align}
	    &\frac{N - (Np_\#+ Np_0)}{N - 1}\notag\\
	    \leq & \frac{N- \frac{1}{K}\sum_{n=1}^N\sum_{k=1}^K p_{(\#)}^{k,n}-\frac{1}{K}\sum_{k=1}^{K}\sum_{\pi\in \cP} p_{(0)}^{k,\pi}}{N-1} \label{D-constraint-proof-1} \\
	    \leq & \frac{\frac{1}{K}\sum_{k=1}^K\left(N- \sum_{n=1}^N p_{(\#)}^{k,n}-\sum_{\pi} p_{(0)}^{k,\pi}\right)}{N-1}\leq D, \label{D-constraint-proof-2}
	\end{align}
	due to the last set of constraints in (P1). 
	
	Therefore, $\{p_\#,p_0,\ldots,p_{K-1}\}$ indeed satisfies the constraints in problem $P2$. It remains to show that this assignment leads to a lower objective function value in $(P2)$ than the optimal value of $(P1)$. For this purpose, we write the inequalities (\ref{permute-reduce-3}-\ref{eqn:maxs}).

\begin{figure*}[!t]
\setcounter{MYtempeqncnt}{\value{equation}}
\hrulefill
		\begin{align}
			&\frac{1}{N}\sum_{n=1}^N 2^{\Lc(M \rightarrow Q^{[M]}_n)}\notag\\
			&= \frac{1}{N}\sum_{n=1}^{N}\left[ \sum_{k=1}^K p_{(\#)}^{k,n}+\max_k\left(\sum_{j\in [1:N]: j\neq n} p_{(\#)}^{k,j}	+\sum_{\pi:\phi_{n}^*(k,(0,\pi))=0}p_{(0)}^{k,\pi}\right)+\sum_{\|q\|\neq 0}\max_{k} \left(\sum_{(f,\pi):\phi_{n}^*(k,(f,\pi))=q}p_{(f)}^{k,\pi}\right)\right]    \label{permute-reduce-3} \\
			&\geq \frac{1}{N}\sum_{n=1}^{N}\sum_{k=1}^K p_{(\#)}^{k,n}+ \max_k\left(\frac{1}{N}\sum_{n=1}^N\left(\sum_{j\in [1:N]: j\neq n} p_{(\#)}^{k,j}+\sum_{\pi:\phi_{n}^*(k,(0,\pi))=0}p_{(0)}^{k,\pi}\right)\right)\notag\\
			&\qquad\qquad \qquad\qquad +\sum_{\|q\|\neq 0}\max_{k} \left(\frac{1}{N}\sum_{n=1}^{N}\left(\sum_{(f,\pi):\phi_{n}^*(k,(f,\pi))=q}p_{(f)}^{k,\pi}\right)\right)   \label{permute-reduce-4} \\
			& = Kp_{\#}+\underbrace{
			\max_k\left(\frac{N-1}{N}\sum_{n=1}^Np_{(\#)}^{k,n}+\frac{1}{N}\sum_{n=1}^N\sum_{\pi:\phi_{n}^*(k,(0,\pi))=0}p_{(0)}^{k,\pi}\right)}_{T_1}\notag\\
			&\qquad+ \underbrace{\sum_{\|q\|=K}\max_{k} \left(\frac{1}{N}\sum_{n=1}^{N}\left(\sum_{(f,\pi):\phi_{n}^*(k,(f,\pi))=q}p_{(f)}^{k,\pi}\right)\right)}_{T_2}+\underbrace{\sum_{j=1}^{K-1}\sum_{\|q\|=j}\max_{k} \left(\frac{1}{N}\sum_{n=1}^{N}\left(\sum_{(f,\pi):\phi_{n}^*(k,(f,\pi))=q}p_{(f)}^{k,\pi}\right)\right)}_{\sum_{j=1}^{K-1}T_3(j)}.\label{eqn:maxs}
			\end{align}
\hrulefill
\setcounter{equation}{\value{MYtempeqncnt}+3}
\end{figure*}

Let us consider the last three terms in (\ref{eqn:maxs}) individually. For the first term, observe that: 
\begin{align}
T_1&\geq \frac{1}{K}\sum_{k=1}^K\left(\frac{N-1}{N}\sum_{n=1}^Np_{(\#)}^{k,n}+\frac{1}{N}\sum_{n=1}^N\sum_{\pi:\phi_{n}^*(k,(0,\pi))=0}p_{(0)}^{k,\pi}\right) \notag \\
&=(N-1)p_{\#}+\frac{1}{KN}\sum_{k=1}^K\sum_{n=1}^N\sum_{\pi:\phi_{n}^*(k,(0,\pi))=0}p_{(0)}^{k,\pi}\\
&=(N-1)p_{\#}+\frac{1}{KN}\sum_{k=1}^K \sum_{\pi\in \cP} p_{(0)}^{k,\pi}\label{eqn:count0}\\
&=(N-1)p_{\#}+p_0,
\end{align}
 where the inequality is due to the convexity of the max function, and (\ref{eqn:count0}) is by the fact that for each $k$, each permutation is counted exactly once in the summation $\sum_{n=1}^N\sum_{\pi:\phi_{n}^*(k,(0,\pi))=0}$. To see that latter, observe that each $\pi$ must map $f=0$ to a query $q=0$ at one and only one of the servers. 
 In a similar manner
 \begin{align}
 T_2 &\geq      	\sum_{\|q\|=K} \frac{1}{K}\sum_{k=1}^{K} \left(\frac{1}{N}\sum_{n=1}^{N}\left(\sum_{(f,\pi):\phi_{n}^*(k,(f,\pi))=q}p_{(f)}^{k,\pi}\right)\right)\notag\\
 &=\frac{1}{KN} \sum_{k=1}^{K}  \sum_{\|q\|=K} \sum_{n=1}^{N} \sum_{(f,\pi):\phi_{n}^*(k,(f,\pi))=q}p_{(f)}^{k,\pi}\\
 &=\frac{N-1}{KN} \sum_{k=1}^{K} \sum_{\|f\|=K-1}\sum_{\pi}  p_{(f)}^{k,\pi}\label{eqn:countK}\\
 &= t_Kp_{K-1}=t_K\max(p_{K-1},p_K),
 \end{align}
 where the inequality is again due to the convexity of the max function.  To see (\ref{eqn:countK}), let us introduce the notation
 \begin{align}
 q|k = (q_1,q_2,\ldots,q_{k-1},q_{k+1},\ldots,q_{K}),
 \end{align}
 i.e., the query vector with the $k$-th symbol removed. Then 
 \begin{align}
 \sum_{n=1}^{N} \sum_{(f,\pi):\phi_{n}^*(k,(f,\pi))=q}p_{(f)}^{k,\pi}= \sum_{\pi}p_{(q|k)}^{k,\pi},
 \end{align}
 because for that fixed $q$, the corresponding $f$ is fixed, and for each $\pi$, there is one and only one $n$ such that $\phi_{n}^*(k,(f,\pi))=q$ holds. Therefore, 
 \begin{align}
&\sum_{\|q\|=K} \sum_{n=1}^{N} \sum_{(f,\pi):\phi_{n}^*(k,(f,\pi))=q}p_{(f)}^{k,\pi}\notag\\
&\qquad= \sum_{\|q\|=K}\sum_{\pi}p_{(q|k)}^{k,\pi}=(N-1)\sum_{\|f\|=K-1}\sum_{\pi}p_{(f)}^{k,\pi},
\end{align} 
because for each $k$, each $f$ with $\|f\|=K-1$ corresponds to exactly $N-1$ queries with $\|q\|=K$. 

For the last term, consider a fixed $j$, and then
	
		\begin{align}
			T_3(j)&\geq \sum_{\|q\|=j}\max_{k}\left(\frac{1}{N}\sum_{n=1}^{N}\left(\sum_{(f,\pi):\phi_{n}^*(k,(f,\pi))=q}p_{(f)}^{k,\pi}\right)\right)  \nonumber \\
			&=\frac{1}{N}\sum_{\|q\|=j}\max_{k} \sum_{\pi}p_{(q|k)}^{k,\pi}\label{eqn:countj}\\
			&\geq \frac{1}{N}\sum_{\|q\|=j}\max\left\{\frac{1}{j}\sum_{k=1}^{K}\mathds{1}(\|q|k\|=j-1)\sum_{\pi} p_{(q|k)}^{k,\pi},\right.  \nonumber  \\
			&\qquad \left.\frac{1}{K-j}\sum_{k=1}^{K}\mathds{1}(\|q|k\|=j)\sum_{\pi} p_{(q|k)}^{k, \pi}\right\}  \label{permute-reduce-9-2}  \\
			&\geq \frac{1}{N}\max\left\{\frac{1}{j}\sum_{k=1}^{K}\sum_{\|q\|=j}\mathds{1}(\|q|k\|=j-1)\sum_{\pi}p_{(q|k)}^{k,\pi},\right.  \nonumber  \\
			&\qquad \left.\frac{1}{K-j}\sum_{k=1}^{K}\sum_{\|q\|=j}\mathds{1}(\|q|k\|=j)\sum_{\pi}p_{(q|k)}^{k,\pi}\right\}  \label{permute-reduce-9-3}  \\
			&= \frac{1}{N}\max\left\{\frac{(N-1)}{j}\sum_{k=1}^{K}\sum_{\|{f}\|=j-1}\sum_{\pi}p_{(f)}^{k,\pi}\right., \nonumber \\
			& \qquad\qquad \left.\frac{1}{K-j}\sum_{k=1}^{K}\sum_{\|f\|=j}\sum_{\pi} p_{(f)}^{k,\pi}\right\}  \label{permute-reduce-9-4}  \\
			&=\frac{1}{N}\max\left\{\frac{t_j}{Ks_{j-1}}\sum_{k=1}^{K}\sum_{\|f\|=j-1}\sum_{\pi} p_{(f)}^{k,\pi}\right., \nonumber \\
			& \qquad\qquad \left.\frac{t_j}{Ks_j}\sum_{k=1}^{K}\sum_{\|f\|=j}\sum_{\pi} p_{(f)}^{k,\pi}\right\}  \label{permute-reduce-9-5}  \\
			&=t_j\cdot \max\left\{p_{j-1},p_j\right\},  \label{permute-reduce-9-6} 
		\end{align}
		where \eqref{permute-reduce-9-2} follows from $|\{f:f=q|k,\|f\|=j-1\}|=j$ and $|\{f:f=q|k,\|f\|=j\}|=K-j$ for a given $q$, and $\mathds{1}(\cdot)$ is the indicator function, 
		\eqref{permute-reduce-9-3} follows from the convexity of the max function, 
		\eqref{permute-reduce-9-4} follows by counting the number of $f$ for each $q$, 
		\eqref{permute-reduce-9-5} follows from $\frac{t_j}{s_{j-1}}=\frac{K(N-1)}{j}$ and $\frac{t_j}{s_j}=\frac{K}{K-j}$, 
		and \eqref{permute-reduce-9-6} follows by the assignment of $p_j$ in \eqref{def-p-j-proof}. This proves the inequality $(P1) \geq (P2)$, and the proof is complete.
  
\end{proof}

\section{Proof of Theorem \ref{theorem:homo}}
\label{sec:proveTh2}

\begin{proof}
The proof follows a similar line as that in \cite{Zhou2020a}, however, with the additional escape $\#$ retrieval patterns, the construction of the dual variables becomes considerably more complex. 

First define $\hat{p}=1-D+D/N$, where $\hat{p}\in [N^{-K},N^{-1}]$. We can rewrite the problem (P2) as a linear program denoted as $P2'$
\begin{mini}[2]
  {\substack{p_\#,p_0,\ldots,\\p_{K-1},m_1,\ldots,m_K}}{\sum_{j=1}^K t_jm_j +p_0+(N+K-1)p_{\#}}{}{}
  \addConstraint{p_\#,p_0, p_1,\ldots,p_{K-1}}{\geq 0}{}
  \addConstraint{N p_{\#}+\sum_{j=0}^{K-1}Ns_j p_j}{=1}{}
  \addConstraint{p_{j-1}-m_j}{\leq 0,\quad\forall j\in [1:K]}{}
  \addConstraint{p_{j}-m_j}{\leq 0,\quad\forall j\in [1:K]}{}
  \addConstraint{p_0+p_\#}{\geq \hat{p}.}{}
\end{mini}

The Lagrangian is
	\begin{align}
		&\mathscr{L} = p_0+(N+K-1)p_{\#} + \sum_{j=1}^K t_j m_j - \sum_{j=0}^{K-1} \eta_j p_j \notag\\
		&\quad-\eta_\# p_\#+ \mu\left(Np_\#+ \sum_{j=0}^{K-1} N s_j p_j -1 \right)   \nonumber \\
		&\quad+ \sum_{j=1}^K \left[ \lambda_j(p_j - m_j) + \mu_j(p_{j-1} - m_j) \right] \notag \\
            &\quad+ \lambda(\hat{p} - p_0-p_\#). 
	\end{align}
	
	Then we can write the KKT conditions as follows:
	\begin{enumerate}
		\item stationarity: 
		\begin{numcases}{}
		      (N-1+K)-\eta_\#+N\mu-\lambda=0\label{eqn:dpsharp}\\
			1 - \eta_0 + N\mu + \mu_1 - \lambda = 0   \label{eqn:dp0}\\
			-\eta_j +N\mu s_j + \lambda_j + \mu_{j+1} = 0, ~ j\in[1:K-1]   \label{eqn:dpj}\\
			t_j -\lambda_j - \mu_j = 0, ~ j\in[1:K],   \label{eqn:dmj}
		\end{numcases}
		
		\item primal feasibility: 
		\begin{numcases}{}
			Np_\#+N\sum_{j=0}^{K-1} s_j p_j -1  = 0 \label{eqn:peq} \\
			0 \leq p_{\#}, 0 \leq p_{j-1},p_{j} \leq m_j,  ~ j\in[1:K]   \label{eqn:pineq}\\
			p_0+p_\# \geq \hat{p},
		\end{numcases}
		
		\item dual feasibility: 
		\begin{numcases}{}
			\lambda \geq 0,~\eta_\#\geq 0 \\
			\lambda_j \geq 0,~\mu_j \geq 0,~ \eta_{j-1} \geq 0,~ j\in[1:K],   \label{eqn:dual} 
		\end{numcases}
		
		\item complementary slackness: 
		\begin{numcases}{}
		       \eta_\# p_\#=0 \label{eqn:spsharp} \\
			\eta_j p_j = 0, ~j\in[1:K-1]    \label{eqn:spj} \\
			\lambda_j(p_j - m_j) =0,~j\in[1:K]    \label{eqn:smj}\\
			\mu_j(p_{j-1} - m_j) = 0,~j\in[1:K]     \\
			\lambda(\hat{p}- p_0-p_\#) = 0.   \label{eqn:sp0}
		\end{numcases}
	\end{enumerate}
	Since the problem is linear and feasible, any solution to the KKT conditions is optimal. Therefore, finding such a solution whose primal variables yield the values in Theorem \ref{theorem:homo} would conclude the proof. We claim that such a solution to the KKT conditions is as follows:
	\begin{enumerate}[1)]
		\item primal variables: 
		\begin{numcases}{}
			p_\#  =  \frac{N^{K}\hat{p}-1}{N^K-N} \label{eqn:psharp_value} \\
			p_j  = \frac{1-Np_\#}{N^{K}},~ j\in[0:K-1]\\
			m_j = p_{j-1}, ~ j\in[1:K], 
		\end{numcases}
		
		\item dual variables: 
		\begin{numcases}{}
		      \eta_\# =0 \\
			\eta_j = 0, ~ j \in[0:K-1]\\
			\lambda = \frac{N^{K-1}(K-1)}{N^{K-1}-1}\\
			\mu  =  \frac{N^{K-2}(K-1)}{N^{K-1}-1}-\frac{N+K-1}{N} \label{def:mu} \\
			\lambda_j  = \sum_{i=0}^{j}t_i-(N+K-1)\sum_{i=0}^{j-1}s_i+ \lambda \sum_{i=1}^{j-1}s_i, \notag\\
			\qquad\qquad\qquad\qquad\qquad\qquad j\in[1:K] \\
			\mu_j =  -\sum_{i=0}^{j-1}t_i+(N+K-1)\sum_{i=0}^{j-1}s_i- \lambda \sum_{i=1}^{j-1}s_i, \notag\\
			\qquad\qquad\qquad\qquad\qquad\qquad j\in[1:K]   
		\end{numcases}
	\end{enumerate}	
	It remains to show that the solution above indeed satisfies all the KKT conditions. First note that with the given assignment, $\lambda_K=0$. With this observation, the conditions of stationarity, primal feasibility, and complementary slackness can be verified by simply plugging in the variable assignments. For dual feasibility, observe that $\lambda\geq0$, and it remains to show that $\lambda_k\geq 0$ and $\mu_k\geq $ for $j=1,2,\ldots,K$. This is established in Lemma \ref{lem:mono} below. By Proposition \ref{prop:DL}, with homogeneous trustworthy $\gamma_1=\ldots=\gamma_n=\gamma$, it is straightforward to verify the primal variable assignment indeed leads to the surrogate leakage
    \begin{align}
         &\rho_{(\text{Max-L})} =\sum^{N}_{n=1}\gamma \left( \sum_{j=1}^K t_jm_j +(p_0+(N-1)p_{\#})+Kp_{\#} \right) \notag \\
         = &N\gamma \left( 1+\frac{(K-1)\left[ N^{K-1}(N-(N-1)D)-1\right]  }{N^{K}-N} \right)
    \end{align}
    Now it remains to show that the allocation in Theorem \ref{theorem:allocations} provides us with the same optimal surrogate maximal leakage given above. By assigning $\hat{p}_{\#} = Np_{\#}$, we have
    \begin{align}
        \rho_{(\text{Max-L})} & = \sum_{n=1}^N \gamma 2^{\Lc(M \rightarrow Q^{[M]}_n)}\\
        & = \gamma\left( N^{K-1}(\frac{1-\hat{p}_{\#}}{N^{K-1}})+K\hat{p}_{\#} \right) \notag \\
        &\qquad + (N-1)\gamma \left( N^{K-1}(\frac{1-\hat{p}_{\#}}{N^{K-1}})+\hat{p}_{\#} \right) \\
        & = \gamma \left( 1+(K-1)\hat{p}_{\#} \right) +(N-1)\gamma \\ 
        & = N\gamma + \gamma(K-1)\hat{p}_{\#} \\
        & =  N \gamma \left(1+\frac{(K-1)\left[ N^{K-1}(N-(N-1)D)-1\right]  }{N^{K}-N} \right)  
    \end{align}
    This completes the proof.
    
	\end{proof}
	
	\begin{lemma}\label{lem:mono}
		The solution given above satisfies
		\begin{align*}
			&\frac{(N+K-1)\sum_{i=0}^{j-1}s_i- \lambda \sum_{i=1}^{j-1}s_i}{\sum_{i=0}^{j-1} t_i } \geq 1 \notag\\
			&\quad\geq \frac{(N+K-1)\sum_{i=0}^{j-1}s_i- \lambda \sum_{i=1}^{j-1}s_i}{\sum_{i=0}^{j} t_i }, ~ j\in[1:K]. 
		\end{align*}
	\end{lemma}
	
        \begin{proof}
            
    It is straightforward to verify that $(N+K-1-\lambda)\geq 0$ whenever $N\geq 2$ and $K\geq 1$. When $j=1$, let $R_{1} = (N+K-1)s_{0} $ and $\hat{R}_{1} = t_0 + t_1 $, the RHS function satisfies	      
	      \begin{align*}
			\frac{R_{1}}{\hat{R}_{1}}=\frac{(N+K-1)}{1+K(N-1)} \leq 1,
		\end{align*}
		when $N\geq 2$ and $K\geq 1$. Next, notice that the function below is monotonically increasing with respect to $j$ for $j\in[2:K] $, because 
		\begin{align}
			&\frac{R_{j}}{\hat{R}_{j}} = \frac{(N+K-1-\lambda)s_{j-1}}{t_{j}} \notag\\
			= &\frac{(N+K-1-\lambda)\binom{K-1}{j-1}(N-1)^{j-1}}{\binom{K}{j}(N-1)^{j} } \notag\\
			= &\frac{(N+K-1-\lambda)j}{K(N-1)},  ~ j\in[2:K].
		\end{align}
        Observe that when $j = K$, 
            \begin{align}
			\frac{R_{K}}{\hat{R}_{K}} = \frac{(N+K-1-\lambda)K}{K(N-1)} \geq 1. 
		\end{align}
        Clearly there exists a $j^{\ast} \in [2:K-1]$ s.t. 
        \begin{align}
            \frac{R_{2}}{\hat{R}_{2} } \leq \cdots \leq \frac{R_{j^{\ast }}}{\hat{R}_{j^{\ast }} } \leq \frac{R_{1}}{\hat{R}_{1} } \leq \frac{R_{j^{\ast }+1}}{\hat{R}_{j^{\ast }+1} } \leq \cdots \leq \frac{R_{K}}{\hat{R}_{K} }. 
        \end{align}
        For $j \leq j^{\ast}$, the RHS inequality holds because 
        \begin{align}
			RHS = \frac{R_{1}+\cdots+R_{j}}{\hat{R}_{1}+\cdots+\hat{R}_{j}} \leq \frac{R_{1}}{\hat{R}_{1} } \leq 1, ~j\in[1:j^{\ast}].
		\end{align}
       On the other hand, for $j > j^{\ast}$, the right-hand function is monotonically increasing with respect to $j$ since
        \begin{align}
            \frac{R_{1}+\cdots+R_{j-1}}{\hat{R}_{1}+\cdots+\hat{R}_{j-1}} \leq \frac{R_{1}+\cdots+R_{j}}{\hat{R}_{1}+\cdots+\hat{R}_{j}} \leq \frac{R_{j}}{\hat{R}_{j} },~j\in[j^{\ast}+1:K].
        \end{align}
	It follows that the RHS inequalities hold also for $j > j^{\ast}$, since the RHS reaches its maximum when $j = K$ and
		\begin{align}
			&\frac{(N+K-1)\sum_{i=0}^{j-1}s_i- \lambda \sum_{i=1}^{j-1}s_i}{\sum_{i=0}^{j} t_i } \notag\\
			\leq & \frac{(N+K-1)\sum_{i=0}^{K-1}s_i- \lambda \sum_{i=1}^{K-1}s_i}{\sum_{i=0}^{K} t_i }= 1, ~j\in[1:K]. 
		\end{align}		
		
	For the LHS, the function below is monotonically decreasing with respect to $j \in [1:K]$, since 
		\begin{align}
			&\frac{L_{j}}{\hat{L}_{j}} = \frac{(N+K-1-\lambda)s_{j-1}}{t_{j-1}} \notag\\
			= &\frac{(N+K-1-\lambda)\binom{K-1}{j-1}(N-1)^{j-1}}{\binom{K}{j-1}(N-1)^{j-1} } \notag\\
			= &\frac{(N+K-1-\lambda)(K-j+1)}{K}, ~ j\in[2:K].
		\end{align}
        and 
            \begin{align}
			&\frac{L_1}{\hat{L}_{1}} = \frac{(N+K-1)s_{0}}{t_0} =(N+K-1) \notag\\
			\geq &\frac{L_2}{\hat{L}_{2}} = \frac{(N+K-1-\lambda)(K-1)}{K}. 
		\end{align}
        The left-hand function is decreasing in $j$ because 
        \begin{align}
            \frac{L_{1}+\cdots +L_{j-1}}{\hat{L}_{1} +\cdots +\hat{L}_{j-1} } \geq \frac{L_{1}+\cdots +L_{j}}{\hat{L}_{1} +\cdots +\hat{L}_{j} } \geq \frac{L_{j}}{\hat{L}_{j} }, ~  j\in[2:K]. 
        \end{align}
	Therefore, the LHS inequalities hold because the function on the left-hand side reaches its minimum when $j = K $ and 
		\begin{align}
		&\frac{(N+K-1)\sum_{i=0}^{j-1}s_i- \lambda \sum_{i=1}^{j-1}s_i}{\sum_{i=0}^{j-1} t_i }\notag\\
		\geq &\frac{(N+K-1)\sum_{i=0}^{K-1}s_i- \lambda \sum_{i=1}^{K-1}s_i}{\sum_{i=0}^{K-1} t_i }\notag\\
		=&\frac{(N+K-1)N^{K-1}-\lambda(N^{K-1}-1)}{N^K-(N-1)^K}\notag\\
		=&\frac{N^K}{N^K-(N-1)^K}>1, ~ \forall j\in[1:K].
		\end{align}
		This completes the proof.
  
        \end{proof}

\section{Proof of Proposition \ref{prop:P3P4}}
\label{proof:P3P4}
\begin{proof}
\noindent\textit{The direction (P3)$\leq$(P4):} Proof of this direction is trivial since $({P4})$ can be viewed as $({P3})$ with additional constraints enforced through (\ref{eqn:reduced}).

\noindent\textit{The direction (P3)$\geq$(P4):} Given an optimal solution in $({P3})$, we can find the assignment of $p_{\# },p_{0},p_{1},\ldots ,p_{K-1}$:
\begin{align}
    p_{\# } &=\frac{1}{NK} \sum^{K}_{k=1} \sum^{N}_{n=1} p^{k,n}_{(\#) },  \label{def-p-sharp} \\
    p_{j} &=\frac{1}{NKs_{j}} \sum^{K}_{k=1} \sum_{\pi \in \Pc} \sum_{f\in \Fc_{j}} p^{k,\pi }_{(f)}, j\in[0:K-1].  \label{def-p-j}
\end{align}

The constraints in problem $({P4})$ can be easily verified using the same nature in the Max-L setting (Eq. \ref{prob-constraint-proof}-\ref{D-constraint-proof-2}). We only need to show that $({P4})$ using this assignment gives a lower bound on the objective function in $({P3})$. We firstly introduce the notation
\begin{align}
    p_{n}(0|k) &\triangleq \sum_{j\in [1:N]:j\neq n} p^{k,j}_{(\# )}+\sum_{\pi :\phi^{\ast }_{n} \left( k,(0,\pi )\right)  =0} p^{k,\pi }_{(0)}, \\
    p_{n}(q|k) &\triangleq \sum_{(f,\pi ):\phi^{\ast }_{n} \left( k,(f,\pi )\right)  =q} p^{k,\pi }_{(f)}
\end{align}

The homogeneous MI leakage can be written as:
\begin{align}
    &\rho_{(\text{MI})} =  \frac{1}{N} \sum^{N}_{n=1} \text{MI} \left( M\rightarrow Q^{[M]}_{n}\right)  \\
    &=  \frac{1}{N} \sum^{N}_{n=1} \sum_{k=1}^K\sum_{q=\# \& q\in \Qc} \frac{p_{n}(q|k)}{K} \notag\\
    &\qquad\qquad\cdot\bigg{\{} H(M) 
     -H(M|Q^{[M]}_{n}=q) \bigg{\}} \\
    &= p_{\# }\log K+ \frac{1}{N} \sum^{N}_{n=1} \sum_{k=1}^K\sum_{\| q\| =0} \frac{p_{n}(0|k)}{K}\notag\\
    &\qquad\qquad\cdot\bigg{\{} H(M)-H(M|Q^{[M]}_{n}=0)\bigg{\}} \notag \\
    &\qquad + \frac{1}{N} \sum^{N}_{n=1} \sum_{k=1}^K\sum_{\| q\| >0} \frac{p_{n}(q|k)}{K} 
     \notag\\
     &\qquad\qquad\cdot\bigg{\{} H(M)-H(M|Q^{[M]}_{n}=q)\bigg{\}} \\
    &\geq p_{\# }\log K \notag\\
    &+\underbrace{ \frac{1}{N} \sum^{N}_{n=1} \sum_{k=1}^K\sum_{\| q\| >0} \frac{p_{n}(q|k)}{K}\bigg{\{} H(M) 
    -H(M|Q^{[M]}_{n}=q)\bigg{\}} }_{\mathcal{T}},
\end{align}
where the second eqaulity is obtained by splitting the queries into those of the type $\#$ and otherwise, and the inequality is because the uniform distribution maximizes the entropy. On the other hand, for the last term $\Tc$, we have
\begin{align}
    \Tc &= \frac{1}{K} \sum^{K}_{j=1} \sum_{\| q\| =j} \bigg{\{} \frac{1}{N} \sum^{N}_{n=1} \sum^{K}_{k=1} p_{n}(q|k)\log \frac{Kp_{n}(q|k)}{\sum^{K}_{k^{\prime }=1} p_{n}(q|k^{\prime })} \bigg{\}} \label{permute-mi-T} \\
    &\geq \frac{1}{K} \sum^{K}_{j=1} \sum_{\| q\| =j} \bigg{\{} \frac{1}{N} \sum^{K}_{k=1} \left(\sum^{N}_{n=1} p_{n}(q|k)\right) \notag \\
    &\qquad \log \left(\frac{K\sum^{N}_{n=1} p_{n}(q|k)}{\sum^{N}_{n=1} \sum^{K}_{k^{\prime }=1} p_{n}(q|k^{\prime })}\right) \bigg{\}} \label{permute-mi-convexity}  \\
    &\geq \frac{1}{K} \sum^{K}_{j=1} \bigg{\{} \frac{1}{N}\sum^{K}_{k=1}\sum_{\| q\| =j} \mathds{1}\left( \| q|k\| =j-1\right) \notag \\
    &\qquad\qquad\qquad \sum_{\pi } p^{k,\pi }_{(q|k)}\log \frac{K\sum_{\pi } p^{k,\pi }_{(q|k)}}{\sum^{K}_{k^{\prime }=1} \sum_{\pi } p^{k,\pi }_{(q|k^{\prime })}} \bigg{\}}  \\
    & \quad + \frac{1}{K} \sum^{K}_{j=1} \bigg{\{} \frac{1}{N}\sum^{K}_{k=1}\sum_{\| q\| =j} \mathds{1}\left( \| q|k\| =j\right) \notag \\
    &\qquad\qquad\qquad \sum_{\pi } p^{k,\pi }_{(q|k)}\log \frac{K\sum_{\pi } p^{k,\pi }_{(q|k)}}{\sum^{K}_{k^{\prime }=1} \sum_{\pi } p^{k,\pi }_{(q|k^{\prime })}} \bigg{\}} \label{permute-mi-1} \\
    & = \frac{1}{NK} \sum^{K}_{j=1} \bigg{\{} (N-1) \sum^{K}_{k=1} \sum_{\| f\| =j-1} \sum_{\pi } p^{k,\pi }_{(f)} \notag \\
    &\qquad\qquad\qquad\qquad\qquad \log \frac{K\sum_{\pi } p^{k,\pi }_{(f)}}{\sum^{K}_{k^{\prime }=1} \sum_{\pi } p^{k,\pi }_{(q|k^{\prime })}} \bigg{\}}  \\
    & \quad + \frac{1}{NK} \sum^{K}_{j=1} \bigg{\{} \sum^{K}_{k=1} \sum_{\| f\| =j} \sum_{\pi } p^{k,\pi }_{(f)} \notag \\
    &\qquad\qquad\qquad\qquad\qquad \log \frac{K\sum_{\pi } p^{k,\pi }_{(f)}}{\sum^{K}_{k^{\prime }=1} \sum_{\pi } p^{k,\pi }_{(q|k^{\prime })}} \bigg{\}}  \\
    & \geq \frac{1}{K} \sum^{K}_{j=1} \bigg{\{} \frac{t_{j}j}{NKs_{j-1}} \sum^{K}_{k=1} \sum_{\| f\| =j-1} \sum_{\pi } p^{k,\pi }_{(f)} \notag \\
    &\qquad\qquad\qquad \log \frac{K\frac{1}{NKs_{j-1}} \sum^{K}_{k=1} \sum_{\pi } \sum_{f} p^{k,\pi }_{(f)}}{\sum^{K}_{k^{\prime }=1} \sum_{\pi } p^{k,\pi }_{(q|k^{\prime })}} \bigg{\}} \\
    & \quad + \frac{1}{K} \sum^{K}_{j=1} \bigg{\{} \frac{t_{j}(K-j)}{NKs_{j}} \sum^{K}_{k=1} \sum_{\| f\| =j} \sum_{\pi } p^{k,\pi }_{(f)} \notag \\
    &\qquad\qquad\qquad \log \frac{K\frac{1}{NKs_{j}} \sum^{K}_{k=1} \sum_{\pi } \sum_{f} p^{k,\pi }_{(f)}}{\sum^{K}_{k^{\prime }=1} \sum_{\pi } p^{k,\pi }_{(q|k^{\prime })}} \bigg{\}}   \label{permute-mi-2} \\
    & = \frac{1}{K}\sum^{K}_{j=1} t_{j} \bigg{\{} jp_{j-1}\log \frac{Kp_{j-1}}{jp_{j-1}+(K-j)p_{j}} \notag \\
    &\qquad\qquad +(K-j)p_{j}\log \frac{Kp_{j}}{jp_{j-1}+(K-j)p_{j}} \bigg{\}} \\
    & =\frac{1}{K} \sum^{K}_{j=1} \begin{pmatrix}K\\ j\end{pmatrix} (N-1)^{j} \notag \\
    &\qquad \bigg{\{} jp_{j-1}\log p_{j-1} +(K-j)p_{j}\log p_{j} \notag \\
    &-\left[ jp_{j-1}+(K-j)p_{j}\right]\log \frac{jp_{j-1}+(K-j)p_{j}}{K}\bigg{\}},
\end{align}
where (\ref{permute-mi-convexity}) follows from the log sum inequality shown in Lemma \ref{lemma-convex-function-MI} below; (\ref{permute-mi-1}) follows from $|\{f:f=q|k,\|f\|=j-1\}|=j$ and $|\{f:f=q|k,\|f\|=j\}|=K-j$ for a given $q$, and $\mathds{1}(\cdot)$ is the indicator function; (\ref{permute-mi-2}) follows from Lemma \ref{lemma-concave-entropy} and counting the number of $f$ for each $q$. This proves the direction $({P3})\geq({P4})$ and completes the proof.

\end{proof}

\begin{lemma}
\label{lemma-convex-function-MI}
(Log sum inequality) The term $\Tc$ given above satisfies (\ref{permute-mi-T})$\geq$(\ref{permute-mi-convexity}).
\end{lemma}

\begin{proof}
Notice that to show (\ref{permute-mi-T})$\geq$(\ref{permute-mi-convexity}), we only need to show for each fixed $k$
\begin{align}
    &\sum^{N}_{n=1} p_{n}(q|k)\log \frac{Kp_{n}(q|k)}{\sum^{K}_{k^{\prime }=1} p_{n}(q|k^{\prime })} \geq \notag \\
    &\left(\sum^{N}_{n=1} p_{n}(q|k)\right)\log \left(\frac{K\sum^{N}_{n=1} p_{n}(q|k)}{\sum^{N}_{n=1} \sum^{K}_{k^{\prime }=1} p_{n}(q|k^{\prime })}\right).
\end{align}
That is to show
\begin{align}
    &\sum^{N}_{n=1} p_{n}(q|k)\log \frac{p_{n}(q|k)}{\sum^{K}_{k^{\prime }=1} p_{n}(q|k^{\prime })} \geq \notag \\
    &\left(\sum^{N}_{n=1} p_{n}(q|k)\right)\log \left(\frac{\sum^{N}_{n=1} p_{n}(q|k)}{\sum^{N}_{n=1} \sum^{K}_{k^{\prime }=1} p_{n}(q|k^{\prime })}\right).
\end{align}

Let $f(t)=t\log t$. Notice that $f(t)$ is strictly convex since $f^{\prime \prime }(t)=\frac{1}{t} >0$. Hence by Jensen's inequality, we have 
\begin{align}
    \sum^{N}_{n=1} \alpha_{n} f(t_{n})\geq  f\left( \sum^{N}_{n=1} \alpha_{n} t_{n}\right),
\end{align}
for $\alpha_{n} \geq 0$, $\sum_{n} \alpha_{n} = 1$. Hence we can obtain the inequality we want by setting
\begin{align}
    \alpha_{n} &=\frac{\sum^{K}_{k^{\prime }=1} p_{n}(q|k^{\prime })}{\sum^{N}_{j=1} \sum^{K}_{k^{\prime }=1} p_{j}(q|k^{\prime })} ,\\
    t_{n} &=\frac{p_{n}(q|k)}{\sum^{K}_{k^{\prime }=1} p_{n}(q|k^{\prime })}.
\end{align}
The proof is thus completed.
\end{proof}

\begin{lemma}
\label{lemma-concave-entropy}
\begin{align*}
    \sum^{N}_{n=1} p_{n}\log \frac{p_{n}}{Z} \geq Np^{\prime }\log \frac{p^{\prime }}{Z},
\end{align*}
where $p^{\prime }=\sum^{N}_{n=1} p_{n}/N$ and $Z$ is a positive constant.
\end{lemma}

\begin{proof}
The LHS is equivalent to 
\begin{align}
    \text{LHS} &=\sum^{N}_{n=1} p_{n}\log \frac{p_{n}}{Z} =\sum^{N}_{n=1} p_{n}\log p_{n}-\sum^{N}_{n=1} p_{n}\log Z \notag \\
    &=\sum^{N}_{n=1} p_{n}\log p_{n}-Np^{\prime }\log Z.
\end{align}
Since uniform distribution maximizes the entropy, we have
\begin{align}
    \sum^{N}_{n=1} p_{n}\log p_{n}-Np^{\prime }\log Z &\geq Np^{\prime }\log p^{\prime }-Np^{\prime }\log Z \notag \\
    &=Np^{\prime }\log \frac{p^{\prime }}{Z} =\text{RHS} 
\end{align}
\end{proof}

\section{Proof of Lemma \ref{lemma-MI-achieve-download}}
\label{proof:lemma-MI-achieve-download}
\begin{proof} Supposing $\boldsymbol{p}=(p_{\#},p_{0},\ldots,p_{K-1})$ is a valid solution, we first let $\hat{p}^{\ast}\triangleq p_{0}+p_{\# }$, $\hat{p}\triangleq 1-D+D/N \leq \hat{p}^{\ast}$, $\alpha=(1-N\hat{p})/(1-N\hat{p}^{\ast })$. We explicitly construct the following new assignment $\boldsymbol{p}'=(p^{\prime}_{\#},p^{\prime}_{0},\ldots,p^{\prime}_{K-1})$:
\begin{align}
    & p^{\prime }_{j}=\alpha p_{j}, \quad j\in[0,K-1],\\
    & p^{\prime }_{\# }=\hat{p} -p^{\prime }_{0},
\end{align}

The download cost constraint is obtained trivially since $p^{\prime }_{\# }+p^{\prime }_{0}=\hat{p}$. The total probability constraint can be verified by
\begin{align}
    Np^{\prime }_{\# }+\sum^{K-1}_{j=0} Ns_{j}p^{\prime }_{j} &=N\left( p^{\prime }_{\# }+p^{\prime }_{0}\right)  +\sum^{K-1}_{j=1} Ns_{j}p^{\prime }_{j}\\
    &=N\hat{p} +\frac{1-N\hat{p} }{1-N\hat{p}^{\ast } } \sum^{K-1}_{j=1} Ns_{j}p_{j} \\
    &=N\hat{p} +\frac{1-N\hat{p} }{1-N\hat{p}^{\ast } } \left( 1-N\hat{p}^{\ast } \right) \notag \\
    &=1.
\end{align}

The MI leakage under this newly constructed assignment is given by
\begin{align}
    &I(\boldsymbol{p}') =p^{\prime }_{\# }\log K+\frac{1}{K} \sum^{K}_{j=1} \begin{pmatrix}K\\ j\end{pmatrix} (N-1)^{j} \notag \\
    &\quad \bigg{\{} j\alpha p_{j-1}\log \alpha p_{j-1}+(K-j)\alpha p_{j}\log \alpha p_{j}  \\
    &\qquad -\left[ j\alpha p_{j-1}+(K-j)\alpha p_{j}\right]  \log \frac{j\alpha p_{j-1}+(K-j)\alpha p_{j}}{K} \bigg{\}} \notag \\
    & =p^{\prime }_{\# }\log K+\frac{\alpha }{K} \sum^{K}_{j=1} \begin{pmatrix}K\\ j\end{pmatrix} (N-1)^{j} \bigg{\{} jp_{j-1}\log \alpha p_{j-1} \notag \\
    &\qquad +(K-j)p_{j}\log \alpha p_{j} -\left[ jp_{j-1}+(K-j)p_{j}\right] \notag \\
    &\qquad \left( \log \alpha +\log \frac{jp_{j-1}+(K-j)p_{j}}{K} \right) \bigg{\}} \\
    & =p^{\prime }_{\# }\log K+\frac{\alpha }{K} \sum^{K}_{j=1} \begin{pmatrix}K\\ j\end{pmatrix} (N-1)^{j} \bigg{\{} jp_{j-1}\log p_{j-1} \notag \\
    &\qquad +(K-j)p_{j}\log p_{j} \\
    &\qquad -\left[ jp_{j-1}+(K-j)p_{j}\right] \log \frac{jp_{j-1}+(K-j)p_{j}}{K} \bigg{\}}  \notag \\
    & =p^{\prime }_{\# }\log K+\alpha \left( I\left( \boldsymbol{p}\right)  -p_{\# }\log K\right) \\
    & =\left( p^{\prime }_{\# }-\alpha p_{\# }\right)  \log K+\alpha I\left( \boldsymbol{p}\right) \\
    & =\left( \hat{p} -\alpha \hat{p}^{\ast } \right)  \log K+\alpha I\left( \boldsymbol{p}\right)  
\end{align}
Therefore, we have
\begin{align}
    &I(\boldsymbol{p}^{\prime })- I(\boldsymbol{p}) =\left( \hat{p} -\alpha \hat{p}^{\ast } \right)  \log K+(\alpha -1) I\left( \boldsymbol{p} \right)  \\
    = &\left( \hat{p} -\frac{1-N\hat{p} }{1-N\hat{p}^{\ast } } \hat{p}^{\ast } \right)  \log K+(\frac{1-N\hat{p} }{1-N\hat{p}^{\ast } } -1) I\left( \boldsymbol{p} \right)  \\
    = &\frac{\hat{p} -\hat{p}^{\ast } }{1-N\hat{p}^{\ast } } \log K+\frac{N\hat{p}^{\ast } -N\hat{p} }{1-N\hat{p}^{\ast } } I\left( \boldsymbol{p} \right)  \\
    = &\frac{N(\hat{p}^{\ast } -\hat{p} )}{1-N\hat{p}^{\ast } } \left( I\left( \mathbf{p} \right)  -\frac{\log K}{N} \right)  \leq 0,
\end{align}
where the inequality follows from the fact that
\begin{align}
    I(\boldsymbol{p}) &\leq p_{\# }\log K+ \frac{1}{N} \sum^{N}_{n=1} \sum_{k=1}^K\sum_{\| q\| >0} \frac{p_{n}(q|k)}{K} \notag \\
    &\qquad\qquad\qquad \bigg{\{} H(M)-H(M|Q^{[M]}_{n}=q) \bigg{\}} \\
    &\leq p_{\# }\log K+\frac{1}{N} \left( 1-Np_{\# }\right)  H\left( M\right)  =\frac{\log K}{N}.
\end{align}
Therefore, $I(\boldsymbol{p}^{\prime })\leq I(\boldsymbol{p})$, and this completes the proof.

\end{proof}

\end{document}